\documentclass[journal,10pt,twocolumn,twoside]{IEEEtran}

\usepackage{amsmath,amssymb,amsfonts,amsthm}

\usepackage{bbm}
\usepackage{dsfont}
\usepackage{url}
\usepackage{algorithmicx}
\usepackage{enumerate}
\usepackage[usenames,dvipsnames]{xcolor}
\usepackage{graphicx}
\usepackage{xcolor}
\graphicspath{figs/}
\usepackage{multirow}

\usepackage[
  ruled,
  vlined,
  linesnumbered,
]{algorithm2e}

\newtheorem{theorem}{Theorem}
\newtheorem{proposition}{Proposition}
\newtheorem{corollary}{Corollary}
\newtheorem{lemma}{Lemma}
\newtheorem*{proof}{Proof}

\newtheorem{definition}{Definition}
\newcommand{\vect}[1]{\boldsymbol{#1}}
\newcommand{\eat}[1]{}

\newcommand{\techrep}[2]{#1}

\usepackage{cite}
\ifCLASSINFOpdf
\else
\fi

\begin{document}

\title{A Game-Theoretic Analysis of Adversarial Classification}

\author{Lemonia~Dritsoula, Patrick~Loiseau, and~John~Musacchio
\thanks{Lemonia Dritsoula was with the University of California, Santa Cruz (UCSC), USA and is now at Google. Patrick Loiseau is with EURECOM, France and MPI-SWS, Germany. John Musacchio is with the Department of Technology Management at UCSC. E-mail: lenia@soe.ucsc.edu, patrick.loiseau@eurecom.fr, johnm@soe.ucsc.edu. This work was supported by AFOSR grant FA9550-09-1-0049 and NSF grants CNS-0910711 and CNS-0953884 and we acknowledge funding from the Alexander von Humboldt Foundation. }
}

\maketitle

\begin{abstract}
Attack detection is usually approached as a classification problem. However, standard classification tools often perform poorly because an adaptive attacker can shape his attacks in response to the algorithm. This has led to the recent interest in developing methods for \emph{adversarial classification}, but to the best of our knowledge, there have been very few prior studies that take into account the attacker's tradeoff between adapting to the classifier being used against him with his desire to maintain the efficacy of his attack. Including this effect is key to derive solutions that perform well in practice.

In this investigation we model the interaction as a game between a defender who chooses a classifier to distinguish between attacks and normal behavior based on a set of observed features and an attacker who chooses his attack features (class 1 data). Normal behavior (class 0 data) is random and exogenous. The attacker's objective balances the benefit from attacks and the cost of being detected while the defender's objective balances the benefit of a correct attack detection and the cost of false alarm. We provide an efficient algorithm to compute all Nash equilibria and a compact characterization of the possible forms of a Nash equilibrium that reveals intuitive messages on how to perform classification in the presence of an attacker. We also explore qualitatively and quantitatively the impact of the non-attacker and underlying parameters on the equilibrium strategies.
\end{abstract}


\IEEEpeerreviewmaketitle

\section{Introduction}

Classification is one of the most used tools from machine learning. In its simplest instance, a classification algorithm trains a model from a set of labeled data samples of two different classes (class 0 and class 1) and then uses this model to predict the class of new data samples. Many classification algorithms (Support Vector Machines, Logistic Regression, Naive Bayes, etc.) were developed over the past decades and successfully used in applications ranging from computer vision to biology or marketing~\cite{Barber12a, Sebe05a, citeulike:1442986}. 

One of the most prominent applications of classification is security \cite{Tsai09a}, where a defender typically aims to classify a system's usage into normal/non-attack (class 0) or malicious/attack (class 1). In this framework, \emph{attacks} can range from spams received at a user's inbox (class 1) that must be distinguished from regular emails (class 0) to more serious attacks such as DoS or malicious infiltrations on a server (class 1) that must be distinguished from benign traffic (class 0). Although standard classification algorithms have been used to perform such tasks relatively successfully for some time, recent experimental studies showed that a smart adaptive attacker can easily shape his attacks to render those algorithms inefficient \cite{Nelson09a, Sommer10outsidethe, Thomas:2013:TFA:2534766.2534784, Wang14a}. This leads to the crucial question of how to perform \emph{adversarial classification}, that is classification in a setting where (part of) the data is not i.i.d. but rather generated by an adversary trying to fool the classifier. 

In the last decade, a significant body of work from the machine learning and security communities appeared on adversarial classification (see Section~\ref{sec:related_work}). The literature defines two types of attacks: poisoning and evasion attacks. In poisoning attacks, the attacker can alter the training set. The adversarial classification literature mainly analyzes the vulnerability of standard algorithms in this case and sometimes proposes more robust algorithms. Most of these solutions, however, are either too simplistic in that they do not account for the attacker's adaptive nature, or too pessimistic because they rely on the worst-case (zero-sum) assumption that the attacker and defender have opposite objectives; leading to sub-optimal performance in both cases. In evasion attacks, the classifier is fixed and the attacker simply attempts to reverse-engineer it to find an attack that suits his goal while being classified in class 0. Interestingly, there, the literature recently proposed to use randomized classifiers as a defense; but without any formalism to justify the need for randomization and to optimize the distribution of classifiers used for defense. Overall, the problem of finding a classifier that works optimally against an adaptive attacker with a realistic (nonzero-sum) objective remains open. 

In parallel, a set of works appeared on applying game theory to security scenarios in order to compute optimal defenses against an adaptive adversary, in particular in the contexts of intrusion detection \cite{alpcan2010network} and defense resource allocation \cite{Tambe11a} (see Section~\ref{sec:related_work}). In most games studied there, there exists a unique Nash equilibrium in mixed strategies, which justifies randomization and provides a framework to compute the defense distribution. However, none of these works studied a game with a payoff that captures the full complexity of the adversarial classification problem. 

In this paper, we use a game-theoretic approach to tackle the question of how to perform classification in an adversarial setting where an adaptive attacker whose objective is not opposite to the defender's generates the data of class 1. Specifically, we propose to model the system as a game between an attacker and a defender, where the attacker chooses his attack patterns (corresponding to class 1 data) and the defender chooses a classification strategy. In our game, normal usage (class 0 data) is random and exogenous. The attacker's objective balances the benefit from attacks and the cost of being detected while the defender's objective balances the benefit of a correct attack detection and the cost of false alarm. 
Then, we give a complete analysis of the game's Nash equilibria that provides both an algorithm to compute all Nash equilibria and a characterization of the possible forms of a Nash equilibrium. Our characterization is very compact and reveals several important and intuitive messages on classification in the presence of an attacker. 
\begin{enumerate}
\item First, we find that the defender must randomize the classification rule. 
\item Second, the defender should mix between classifiers only from a small subset of simple classifiers that correspond to applying a threshold on the attacker's benefit from the attack. If the attack pattern is a multi-dimensional quantity (i.e., the classification is based on several features), this is in contrast with known algorithms such as logistic regression which have a predefined shape of the boundary independently of the attacker's goal. We also find that the weight assigned to each threshold classifier is mainly proportional to the marginal reward increase at that point. 
\item Third, the attacker should essentially mimic the distribution of the normal behavior but only on a subset of the support with patterns that yield the highest payoff.  
\item Finally, our results allow us to analyze the investment tradeoffs that a strategic defender is facing when she needs to decide whether or not to acquire more data about the attacker (e.g., investing in a new sensor).
\end{enumerate}
We provide numerical experiments that exemplify these results and their applications. 
Overall, our results provide a first step towards building classification algorithms that work well in general adversarial settings.

\subsection{Related Work}
\label{sec:related_work}

\paragraph{Adversarial classification}

The problem of adversarial classification has inspired a lot of research from the machine learning and security communities (see surveys in \cite{Barreno10a, Huang11a}). 

In a seminal paper, Dalvi et al.~\cite{Dalvi04a} tackle the problem of classifying a malicious intruder in the presence of an innocent user using a setting close to ours where the malicious intruder can perturb his behavior to confuse the classifier. However, they compute only the best response, which means that each player can adapt only once. Instead, we consider a fully adaptive attacker and defender and compute the Nash equilibrium. A number of papers (Globerson and Roweis \cite{Globerson06a}, Zhou et al. \cite{Zhou12a, Zhou14a}) study a similar type of attacks and attempt to propose robust classifiers but always using a worst-case assumption which is overly pessimistic in practice. 

To the best of our knowledge, the only works that analyze adversarial classification using nonzero-sum games are the studies by Br\"{u}ckner et al. \cite{Bruckner09a, Bruckner11a, Bruckner12a}, and the studies by Lis\'y et al. \cite{Lisy14a} and Samusevich \cite{Samusevich16a}. The setting of Br\"{u}ckner et al. \cite{Bruckner09a, Bruckner11a, Bruckner12a}, however, is quite different from ours. They restrict the classifier to particular forms (e.g., logistic regression, where the defender's choice is the set of weights) and identify conditions under which a unique pure Nash equilibrium exists. In contrast, we do not restrict the classifiers a priori and we derive the set of classifiers used at equilibrium (which, as we show, depends on the attacker's utility and can therefore not be fixed without considering this); and we focus on mixed strategy Nash equilibria which are the only ones that make sense in most instances of our model. The setting of Lis\'y et al. \cite{Lisy14a} is closer to ours. They assume that the attacker and the defender both choose a threshold in $[0, 1]$ and that the attacker is detected if his chosen threshold is higher than the defender's. Using nonzero-sum payoffs, they give several properties of the Nash equilibrium (in mixed strategies) and propose a method based on discretization to compute an approximate equilibrium. Samusevich \cite{Samusevich16a} derives extensions based on the same model, notably to take into account multiple types of attackers and bounded rationality. This model, however, abstracts away the complexity of multi-feature classification by assuming that players choose a threshold and that the classification result is computed through a ROC curve. In contrast, our model starts with all possible classifiers in the multi-feature case and shows that well-defined threshold classifiers are sufficient. Also, using a discrete model, we propose an efficient procedure to compute the exact Nash equilibrium. 

All the aforementioned papers study poisoning attacks where the attacker can alter the training set. The study of evasion attacks was pioneered by Lowd and Meek \cite{Lowd05a} in the case of linear classifiers and extended by Nelson et al. \cite{Nelson10a} for convex-inducing classifiers. A recent study by Vorobeychik and Li~\cite{Vorobeychik14a, Li15a} recognizes the inefficiency of deterministic classifiers in this context. The authors compare the efficiency between two classifiers and investigate under which assumptions it is better to deterministically select a single classifier or uniformly randomize between them. In this paper, we formally justify the need for randomization using game theory and we consider a more general form of interaction (nonzero-sum game). We also start by considering an exhaustive set of all possible classifiers and analytically derive the subset of classifiers that the defender uses (threshold classifiers on the attacker's reward) and the distribution on this subset. 

Some forms of thresholds are implemented ad hoc in most spam filtering algorithms to balance false alarm costs and detection gains, using standard deterministic classifiers. Our work provides a formal framework that takes into account the existence of strategic attackers willing to change how they attack to evade detection. Then, randomized classifiers are derived as an optimal and stable solution; no player has an economic benefit to unilaterally deviate. Our analysis shows that the defender should use threshold classifiers on the attack reward, contrasting with known algorithms like logistic regression which has a predefined shape of boundary independent of the attacker's goal.  

In parallel to theoretical research aimed at proposing robust classifiers, several empirical works have studied the vulnerability of standard machine learning algorithms to various kinds of attacks \cite{Nelson09a, Sommer10outsidethe, Thomas:2013:TFA:2534766.2534784, Wang14a}. In particular, Sommer and Paxson~\cite{Sommer10outsidethe} identified reasons why machine learning algorithms do not work well in practical adversarial settings and found that the most common pitfall is that attackers adjust their activity in practice to avoid detection.

\paragraph{Game theory and intrusion detection}

Researchers in the ``game theory for security'' community have also tackled the problem of detecting an attacker using game models, see surveys in \cite{Roy10a, Manshaei13a, Lunt, Lye05a, Christin11_NetworkSecurityGamesCombiningGameTheoryBehavioralEconomics, alpcan2010network}. 

Chen and Leneutre~\cite{DBLP:journals/tifs/ChenL09} address the intrusion detection problem in heterogeneous networks consisting of nodes with different non-correlated security assets. Our model is similar in that different attack vectors can be thought of as distributions over targets of different values. The authors of \cite{DBLP:journals/tifs/ChenL09}, however, assume that the detection probability is the same for any attack and defense strategies, yielding the utility function to be the sum of the utilities on each target. On the contrary, we involve different detection and false alarm rates, which makes the set-up more realistic for adversarial classification, but makes the analysis more challenging. We also assume asymmetric information, since the defender is not aware of whether she faces an attacker or a non-attacker. Other works on game theory for intrusion detection (Alpcan and Ba\c{s}ar \cite{Alpcan03a}, Liu et al. \cite{Liu06a}) have similar restrictive assumptions as \cite{DBLP:journals/tifs/ChenL09} that differentiate them from our work. Lye and Wing~\cite{Lye05a} have also investigated a security problem with multiple targets. They analyze a stochastic game between an attacker and administrator and compute the best response strategies of the players using a non-linear program, while we analyze the Nash equilibria of a classification game, which is a stronger solution concept.

The games mentioned above are nonzero-sum games. A number of other nonzero-sum games have been discussed in the literature under the term ``security games'', see \cite{Korzhyk11a}; but they share the restriction of \cite{DBLP:journals/tifs/ChenL09} that the payoff is the sum of the payoffs on each targets, which does not model the problem of adversarial classification well. Let us finally mention that many applications of nonzero-sum games consider Stackelberg equilibria in which the defender chooses first its mixed action \cite{Korzhyk11a, Tambe11a}. In our model, we consider a simultaneous move game and find the Nash equilibria. We find that our nonzero-sum game is strategically equivalent to a zero-sum game, and thus we can analyze our nonzero-sum game using a zero-sum game.

Barni and Tondi~\cite{Barni13a, Barni14a} use a game-theoretic approach to solve a problem of source identification. Their setting differs from ours in that the attacker chooses a distribution of attacks on top of which randomization is exogenously added, whereas we suppose that the attacker directly chooses his attack vector (using a mixed strategy). More importantly, they consider a zero-sum game and derive an asymptotic Nash equilibrium with numerical computations, while we focus on a nonzero-sum game and derive closed-form expressions for the equilibria. Stamm et al. \cite{Stamm12a} also use game theory to find optimal defense strategies in the context of multimedia forensics. They only derive best responses for the defender given a fixed attack distribution though, and do not characterize the Nash equilibria. 

The classification game we investigate is similar in nature to the inspection game, a multi-stage game between a customs inspector and a smuggler, proposed and studied by Dresher~\cite{Dresher1962} and Maschler~\cite{Maschler1966}. Avenhaus et al.~\cite{Inspection_games2002} find the equilibrium of the general nonzero-sum game by using an auxiliary zero-sum game in which the inspectee chooses a violation procedure and the inspector chooses a statistical test with a given false alarm probability. We do not separate the general nonzero-sum game into two games but show the equivalence to a zero-sum game and provide structure to the equilibrium strategies of a single-shot simultaneous-move game.

A subset of our results appeared in earlier versions of this paper~\cite{DLM_cdc_2012, DLM_gamesec_2012}, in a simplified setting. Indeed, in \cite{DLM_cdc_2012, DLM_gamesec_2012} we consider that the attack is characterized by a scalar quantity rather than a vector and impose the use of threshold strategies for the defender. Instead, here, we allow more flexibility by considering attack vectors of arbitrary dimension and considering the set of all possible classifiers and we prove that, at Nash equilibrium, the defender uses only threshold strategies on the attacker's reward (which is one of our main result). We also extend \cite{DLM_cdc_2012, DLM_gamesec_2012} by considering all possible Nash equilibria and associated defender/attacker strategy form, by identifying when the Nash equilibrium is unique, and by showing numerically how our results can be applied to analyzing the variations of the equilibrium strategies with the model parameters and the trade-off in investing in a new sensor to increase the features used for classification. 

The remainder of the paper is organized as follows. In Section~\ref{sec:game_model}, the classification game we consider is presented.  In Section~\ref{sec:thresholds} we justify the selection of threshold strategies for the defender and the proportionality of the attacker's equilibrium strategy to the non-attacker's distribution (Theorem~\ref{threshold_thm}). Section~\ref{sec:analysis}  provides a Nash equilibrium analysis that gives insights on the structure of the players' equilibrium strategies (Theorem~\ref{aggregate_theorem}). Section~\ref{sec:sims} contains the simulation results, and the paper concludes with remarks in Section~\ref{sec:concl}. To improve the flow of the paper, longer or more technical proofs are relegated \techrep{to the Appendix}{to our technical report \cite{techreport}}.

\vspace{0.05in}
\noindent \textbf{Notational Conventions:}
Throughout the paper, we will use the following notational conventions.
All vectors are assumed to be
column vectors and are denoted by bold lowercase letters (e.g., $\vect{\alpha}$, $\vect{\beta}$).
The inner product of two vectors $\vect
\alpha$, $\vect \beta$ of size $n$ is denoted by $\vect \alpha \cdot \vect \beta =
\sum_{i =1}^n \alpha_i \beta_i$.
 For matrices we use capital greek
letters (e.g., $\Lambda$). We denote matrix elements in row $i$ and column $j$ by $\Lambda(i,j)$.
 We use the \textit{prime} sign (${}^{\prime}$) for
transpose of matrices and vectors. 
We use ``$\min \vect \alpha
$" to denote the minimum element of a vector $\vect \alpha$, and
``$\text{minimize}$" when we minimize a specific expression over some
constraints.
 The indicator function is denoted by
$\mathds{1}_{\text{cond}}$; it is equal to 1 if ``cond" holds and is
equal to 0 otherwise. The column vector of ones of length $N$ is
denoted by $\mathbf{1}_{N}$ and the matrix of ones of dimensions $N
\times M$ is denoted by $1_{N \times M}$. The norm of a vector
$\vect{x}$ of length $N$, denoted by $\|\vect x\|$, always refers to
the $L_1$-norm, i.e, $\|\vect{x}\| = |x_1| + |x_2| + \ldots + |x_N|$,
while the cardinality of a set $\mathcal{S}$ is denoted by $|\mathcal{S}|$.
The probability given to strategy $s$ is denoted by $\alpha_s$.
\vspace{0.05in}

\section{The Classification Game}
\label{sec:game_model}

In this section, we present our game-theoretic model of adversarial classification. 

We consider a strategic situation between a \emph{defender} and an agent that may either be an \emph{attacker} with probability $p$ or a \emph{non-attacker} with probability $1-p$. Throughout the paper, we use the generic term non-attacker to designate a ``normal user'' of the system at stake, that is a non-strategic user without any particular adversarial objective (we will specify later the behavior of a non-attacker). The defender seeks to classify the agent as class 0 (non-attacker) or class 1 (attacker). The strategic attacker seeks to exploit the uncertainty of the defender (about the type of the agent) by attacking in such a way as to avoid being classified as an attacker. This scenario encompasses many applications. For instance, in spam classification, the spammer (the attacker) might change the frequency of words included in an email to evade spam filters. In that case, the non-attacker is the sender of a regular email. In another setting, the owner of fraudulent twitter accounts (the attacker) might try to acquire more followers or publish more posts, so that he will get misclassified as a normal user (the non-attacker here)~\cite{Thomas:2013:TFA:2534766.2534784,stringhini13followers}.

Formally, the agent selects an \emph{attack vector} $v$ in $\mathcal{V}$, where $\mathcal{V}$ is the set of all possible attack vectors. The attack vector contains all features used by the defender for classification; e.g., average number of followers, number of retweets, and others (in social networks fraud), number of initiated connections (in portscanner detection~\cite{Jung04a}), path on a graph among nodes in a network or number of accesses to different targets (in intrusion detection). Therefore, although we use the term ``attack vector,'' our setting is not restricted to pure attack scenarios but instead covers any setting in which a defender needs to detect a malicious user, who seeks to evade detection by gaming. 
The defender selects a \emph{classifier} in $\mathcal{C}$, where $\mathcal{C} \subseteq 2^{|\mathcal{V}|}$ is the set of all possible classifiers. A classifier corresponds to a classification rule that determines the class to which the agent is assigned upon observing his attack vector:
\begin{definition}[Classifier]
A classifier $c$ is a function
$c: \mathcal{V} \rightarrow \{ 0,
1\}$, with 
\begin{equation*}
c(v) =
\left\{
	\begin{array}{ll}
		1  & \mbox{for ``attacker" classification}, \\
		0 & \mbox{for ``non-attacker" classification.} 
	\end{array}
\right.
\end{equation*}
\end{definition} 
We assume that $\mathcal{V}$ is finite, thus $\mathcal{C}$ is also a finite set.

If the agent is a non-attacker, he picks $v \in \mathcal{V}$ according to a distribution  $P_\text{N}(\cdot)$ over $\mathcal{V}$, known to both players. 
If the agent is an attacker, the choice $v$ is strategic: the attacker seeks to maximize the payoff function
\begin{equation}
\label{eq.UA}
U^\text{A}(v, c) = R(v) - c_\text{d}  \mathds{1}_ {c(v)=1},
\end{equation}
where $R: \mathcal{V} \rightarrow  \mathbb{R}_+$ is the reward function, and $c_\text{d}$ is the cost in case of detection.  
We refer to $R(v)$ as the ``reward" (to the attacker) for the attack vector $v$, which is granted to the attacker even in case of detection. In contrast, his ``payoff" is the reward minus the cost if detected.
One can view the reward term as the immediate benefit from the attack, while the detection cost can be interpreted as the attacker's ``opportunity cost'' from losing the opportunity to extract value in future attacks after having been exposed. There are many real world scenarios in which the reward is granted to the attacker even upon detection. In most such scenarios, classification does not occur real-time leaving some exploitation window for the attacker. For instance, consider online app stores that contain malware. By the time these apps are classified as malicious and removed from the stores, the owners of the apps (attackers) will already have benefited, e.g., from online purchases, from compromising users' information, paid ads, and others.

The defender's payoff has two additive components. The first is the expected loss to the attacker. When the attacker is present, the loss to the defender is assumed to be minus the gain of the attacker, $-U^\text{A}(v,c)$. Recall that the attacker's utility $U^\text{A}(v,c)$ is composed of a reward term $R(v)$ minus a detection cost term. Thus the defender is in a sense earning a detection reward that matches the attacker's detection cost. This ``reward'' can be seen as the future costs the defender avoids by detecting the attacker now. Since the defender interacts with an attacker with chance $p$, the expected loss to the attacker is $-pU^\text{A}(v,c)$. The second component captures the expected loss due to false alarms. Since the non-attacker is present with chance $1-p$, the expected false alarm cost is $1-p$ times the chance that a non-attacker would pick a $v$ that gets classified as an attacker. Finally, the whole payoff function is scaled by the constant $1/p$ for the convenience of having the term $U^\text{A}(v,c)$ appear unscaled in the payoff. Note that scaling a player's payoff function by a constant has no strategic effect on a game since the player retains the same preferences among outcomes. The resulting payoff function is

\begin{equation}
U^\text{D}(v, c) = - U^\text{A}(v,c) - \displaystyle \frac{1-p}{p} c_\text{fa} \sum_{v' \in \mathcal{V}} P_N(v') \mathds{1}_ {c(v')=1},
\label{eq.UDhat}
\end{equation}
where $c_\text{fa}$ is a constant that captures the cost of false alarms.
Note that for simplicity of the exposition, we assume that the defender has a cost $R(v)$ (in the $U^\text{A}(v,c)$ term) upon attack vector $v$, equal to the reward to the attacker. Our results, however, can be straightforwardly extended to the case where an attack vector $v$ yields a cost to the defender $D(v)\neq R(v)$. Indeed, this would correspond to a payoff $\hat{U}^\text{D}(v, c)  = U^\text{D}(v, c)  + R(v) - D(v)$ which has no strategic effect since, for any $v \in \mathcal{V}$, maximizing  $\hat{U}^\text{D}$ with respect to $c$ is equivalent to maximizing $U^\text{D}$. Note though, that the payoff of the defender will be changed.

To summarize our model, we define the following game.
\begin{definition}[Classification game $\mathcal{G}$]
The classification game $\mathcal{G} = (\mathcal{V}, \mathcal{C}, p, c_\text{d}, c_\text{fa}, P_\text{N}, R(\cdot))$ is the two-player game between the attacker and the defender, where the strategy space of the attacker is $\mathcal{V}$, the strategy space of the defender is $\mathcal{C} \subseteq 2^{\mathcal{V}}$, and the payoffs are given by \eqref{eq.UA} and \eqref{eq.UDhat}, parameterized by
\begin{itemize}
\item $p \in [0,1]:$ probability that the agent is an attacker;
\item $c_\text{d} \in \mathbb{R}_+:$ cost of detection;
\item $c_\text{fa} \in \mathbb{R}_+:$ cost of false alarm;
\item $P_\text{N}:\mathcal{V} \rightarrow [0,1]:$ probability measure that describes the non-attacker's distribution on $\mathcal{V}$;
\item $R: \mathcal{V} \rightarrow  \mathbb{R}_+:$ the reward function.
\end{itemize}
\end{definition}
Throughout the paper, we analyze this game as a simultaneous, complete information game. In particular we assume that the payoffs' parameters are common knowledge. We are also able to formulate our game in normal (matrix) form since the non-attacker is not strategic. At time 0, the attacker selects the attack vector and the defender selects the classifier. The uncertainty is coming from the presence or not of a strategic player. The expected utility of the defender incorporates this uncertainty through the probability $p$. The game would differ significantly if we supposed that there were two or more different types of strategic attackers. In such a setting we would not be able to reduce our model to a matrix game.

\subsection{Model discussion and limitations} 

Before analyzing the game defined above, let us note that, as every model, it makes simplifying assumptions that limit its applicability and there are interesting aspects of security problems that it does not capture. Nevertheless, through its simplicity, we believe our model provides useful insights about the structure of Nash equilibria and the driving factors and intuition behind the players' equilibrium strategies that can be useful in some application scenarios. These insights would be difficult to obtain through a more complex, mathematically less tractable model. Moreover, the analysis of our simple model yields original and non-trivial mathematical results and derivations that can be useful to study more complex models. We leave this as future work but we clarify here the main limitations of our present model, the scenarios in which our assumptions may hold (or not) and some of the ways in which our results could be useful as a basis for further studies.

One of the main assumptions in our model is that the attacker is granted the reward $R(v)$ even in case of detection, and that the detection cost is the same for the attacker and for the defender. This assumption is technically important as we will see because it gives the ``almost zero-sum'' nature of the game that supports our analysis; but it limits the potential applicability of our model. For the reader to get a better idea of the implication of this assumption, we provide here three example scenarios, two where the assumption is reasonable and one where it is not. 

Consider first the task of classifying an incoming email as spam or not. If the email is marked as spam, then the user never sees it so the attacker is granted the reward only upon no detection. Such a scenario cannot be captured by our model, but there are other examples in which the assumption make sense, such as the following two.
Consider the problem of classifying fake news. There is an immediate reward to the publisher who spreads out the fake news (popularity, re-shares, advertising money). If the publisher gets classified as spreading out fake news, then the publisher loses credibility and/or ability to post more news. The detection cost can be seen as the lost opportunity cost from future posts of news (fake or real) which is also the detection bonus for the defender. Finally, consider the problem of classifying malware in apps. Most apps nowadays monetize ads shown to users but malicious apps can trick users to click on ads or gain important personal information of the user who has installed the app. When a malicious app is detected, some harm has already been done that cannot be reversed. The detection cost for the attacker would be the lost opportunity cost from being banned to show future ads. This cost can be seen as the detection bonus for the defender because in the place of this bad actor, another legitimate app can show ads with some expected lifetime reward. (Essentially, the impressions that would be allocated to this app can be diverted to be shown to other similar apps.)

Our model also assumes that parameter $c_d$ is a constant independent of $v$, i.e., the defender gets a constant bonus independent of the attack for detecting the attack. This assumption is instrumental in obtaining the results that the defender only mixes amongst threshold strategies on the attacker's reward. Again, this assumption is reasonable in many cases but not in others. For the fake news or malware apps scenarios above, it is reasonable to assume that this cost is independent from the actual attack intensity, since it can be seen as lost opportunity cost from monetizing ads in the future. For cases such as audits from IRS though, there are usually fines imposed to tax fraudsters which makes the cost symmetric (whatever the fraudster pays, the IRS wins), but the fine usually depends on the amount misreported. Building on our results for the simpler case, however, we believe that it would be possible to study the game where $c_d$ depends on $v$. One would need to find a ranking of attack vectors (depending on both functions $R$ and $c_d$) such that an equivalent of our Lemma~\ref{monotone_pd} holds. 

Finally, our model is specified as a strategic-form game in which players choose their actions once, without knowing how their opponent has played. While this is a simplification of reality, the framework does not prevent us from considering players that take into account the future ramifications of the immediate outcome. For instance, attackers may perceive an opportunity cost of being less able to profit from future attacks by being detected now. As in models of economic competition and market entry, players can incorporate these opportunity costs into their payoff functions that describe their preferences over the outcomes of the ``one-shot" game at hand~\cite{Tirole_IO}. That said, there are phenomena in repeated strategic interactions that are exposed only when the repeated interaction is modeled explicitly, such as in repeated games in which cooperative outcomes are enforced by the threat of future punishment~\cite{Tirole}. If one were to develop and analyze an adversarial classification game that explicitly models the dynamics of an attacker and defender interacting on multiple occasions over an extended period, that model would need the analysis of the one-shot game to construct the needed value functions (see e.g., \cite{Forges92a}). Thus our present results, in particular through the exponential reduction of the strategy space of the defender (Theorem~\ref{threshold_thm}) can be seen as a stepping-stone for constructing such a dynamic game model.

\section{Justification of threshold strategies}
\label{sec:thresholds}

In this section, we show that in equilibrium, the defender's strategy space can be reduced to threshold classifiers on the attacker reward. The sufficiency of threshold classifiers is a useful result as it allows us to compute the Nash equilibria of the game analytically and efficiently while simultaneously providing intuition about the players' equilibrium strategies. In particular, the number of classifiers can be as large as $2^{|\mathcal{V}|}$ while the number of threshold classifiers is of size $|\mathcal{V}|+1$, so this is a great reduction in the defender's strategy space. We first introduce a number of useful definitions for the analysis of our game. Then in Section~\ref{sec.reduction-defender}, we show that the defender's strategy space can be reduced to the set of threshold classifiers and in Section~\ref{sec.reduction-attacker}, we show that the attacker's strategy space can also be reduced.

We will be interested in mixed strategy equilibria, in which the attacker randomizes across multiple attack vectors with a probability distribution $\vect \alpha$ on $\mathcal{V}$ and the defender randomizes across multiple classifiers with a distribution $\vect \beta$ on $\mathcal{C}$. 
The expected attacker and defender payoffs are then given by 
  \begin{equation}
  U^\text{A}(\vect \alpha, \vect \beta) = \sum_{v \in \mathcal{V}} \sum_{c \in
  \mathcal{C}} \alpha_v U^\text{A}(v, c) \beta_c,
  \label{eq:attackerua}
  \end{equation}
  \begin{equation}
U^\text{D}(\vect \alpha, \vect \beta) = \sum_{v \in \mathcal{V}} \sum_{c \in
  \mathcal{C}} \alpha_v U^\text{D}(v, c) \beta_c.
  \label{eq:defenderud}
\end{equation}
Note that a pure strategy is a special case of mixed strategies in which that particular pure strategy is selected with probability $1$ and every other strategy with probability $0$.\footnote{For most instances of interest no pure-strategy equilibrium exists. If one player chooses deterministically a pure strategy, the opponent would switch to a strategy to either evade detection completely or to guarantee detecting the attacker.} 

Recall the definition of Nash equilibrium~\cite{Tirole}:
\begin{definition}[Nash equilibrium]
The pair of probability measures $(\vect \alpha, \vect \beta)$ on $\mathcal{V}$ and $\mathcal{C}$ respectively is a Nash equilibrium (NE) of game $\mathcal{G}$ if each player's mixed strategy is a best response to the other player's mixed strategy, i.e., 
\begin{equation}
U^\text{A}(\vect \alpha, \vect \beta) \geq U^\text{A}(\vect{\hat{\alpha}}, \vect \beta), \label{eq:ua}
\end{equation}
for every probability distribution $\vect{\hat{\alpha}}$ over $\mathcal{V}$, and
\begin{equation}
U^\text{D}(\vect \alpha, \vect \beta) \geq U^\text{D}(\vect \alpha, \vect{\hat{\beta}}) \label{eq:ub}
\end{equation}
for every probability distribution  $\vect{\hat{\beta}}$  over $\mathcal{C}$.
\end{definition}

We define the notion of best-response equivalent games in the same way as in \cite{Rosenthal_1974}:

\begin{definition}
Two games are {\bf best-response equivalent} if the sets of best response strategies of a player in both games coincide for any 
strategy of the other player. 
\eat{\commentlenia{}{Use it somewhere? If not, omit it!}
Two games $\mathcal{G}_1$ and $\mathcal{G}_2$ are \textbf{payoff equivalent} iff the set of payoff pairs (payoffs realized by the attacker and the defender) in $\mathcal{G}_1$, coincides with that of $\mathcal{G}_2$ in the payoff space.}
\label{equivalence}
\end{definition}
Note that for best-response equivalent games, the strategy spaces of each player in both games need to be the same. 

We now define the reduced strategy space for the defender that consists of threshold rules on all possible attack rewards. 
\begin{definition}[Set of threshold classifiers]
\label{def.CT}
$$\mathcal{C}^T = \{c \in \mathcal{C} : c(v) =\mathds{1}_{R(v) \ge t}, \forall v \in \mathcal{V}  \text{ for some } t \in \mathbb{R}\}.$$
\end{definition}
When using a threshold classifier, the defender compares what the attack reward would have been from the observed attack vector to a threshold instead of computing a mapping from any possible attack vector to a detection probability.  
We assume that $\mathcal{C}^T \subseteq \mathcal{C}$, which holds for any reasonable $\mathcal{C}$, in particular for $\mathcal{C}=2^{\mathcal{V}}$. The outcome of the analysis of the paper is that in most cases the defender should not use a single optimal value of $t$. Instead, the threshold should be chosen randomly from a particular range of values with a particular distribution. This range and distribution are found using the analysis of this paper.
   
We also define the probability of detection function as the probability of class 1 classification (or detection) given the attack vector $v$ and the defender's strategy $\vect \beta$. 
\begin{definition}[Probability of detection function]
The probability of detection for an attack vector $v$ and defender's strategy $\vect \beta$ is defined
as \begin{equation} \pi^{\vect \beta}_\text{d}(v) = \sum_{c\in \mathcal{C}} \beta_c \mathds{1}_
{c(v)=1}, \quad \forall v\in \mathcal{V} \label{eq:prob_det}. \end{equation}
\end{definition}

\subsection{Defender's reduced strategy space}\label{sec.reduction-defender}

Threshold strategies are simple and intuitive, but their optimality is not guaranteed a priori. Yet, in the classification game, we show that threshold strategies are optimal in a sense formalized in the following theorem. 
\begin{theorem}
For any NE  ($\vect \alpha, \vect \beta$) of $\mathcal{G} = (\mathcal{V}, \mathcal{C}, p, c_\text{d}, c_\text{fa}, P_\text{N}, R(\cdot))$, there exists a NE of $\mathcal{G}^T =(\mathcal{V}, \mathcal{C}^T, p, c_\text{d}, c_\text{fa}, P_\text{N}, R(\cdot))$ with
the same $\vect \alpha$ and equilibrium payoff pair and the same $\pi_\text{d}$ in the support of the non-attacker's distribution.
\label{threshold_thm}
\end{theorem}
Theorem~\ref{threshold_thm} shows in particular that, when restricted to using only threshold classifiers, the defender achieves the same equilibrium payoff. Hence, although there may exist Nash equilibria where the defender uses other classifiers, he does not lose anything by using only threshold classifiers.

At a high level, the intuition behind Theorem~\ref{threshold_thm} is as follows. First, we can show that the utilities depend on the defender's strategy $\vect \beta$ only through the probability of detection function $\pi^{\vect \beta}_\text{d}$. Then, we show that, at any NE, $\pi^{\vect \beta}_\text{d}$ is non-decreasing in the attacker's reward (i.e., higher rewarding vectors have a higher probability of being detected). Finally, we show that any probability of detection function that is non-decreasing in the attacker's reward can be achieved by a mix of threshold classifiers.
The proof of Theorem~\ref{threshold_thm} is provided at the end of this section. We first establish a series of lemmas that give information about the game structure and will be used in the proof of Theorem~\ref{threshold_thm}. 

\begin{lemma}
For any strategy profile $(\vect \alpha, \vect \beta)$ of $\mathcal{G} =
(\mathcal{V},\mathcal{C}, p, c_\text{d}, c_\text{fa}, P_\text{N}, R(\cdot))$,
the expected payoffs of the players depend on $\vect \beta$ only through the probability of detection function 
$\pi^{\vect \beta}_\text{d}(\cdot)$:

\begin{align}
\!\!U^\text{A}({\vect \alpha, \vect \beta})& \! = \! \sum_{v \in \mathcal{V}} \Big(\alpha_v R(v) - c_d \alpha_v
\pi^{\vect \beta}_\text{d}(v) \Big) \label{eq:u_na} ,\\
\!\!U^\text{D}({\vect \alpha, \vect \beta})& \! = \! - U^\text{A}({\vect \alpha, \vect \beta}) \! -\!  \displaystyle  \frac{1-p}{p} c_{\text{fa}} \sum_{v^{\prime} \in \mathcal{V}} \!\! \left(
P_\text{N}(v^{\prime}) \pi^{\vect \beta}_\text{d}(v^{\prime}) \right) \label{eq:u_nd}.
\end{align}
 \label{same_pd}
\end{lemma}

By abuse of notation, we denote the probability of detection function by $\pi_\text{d}$, instead of  $\pi^{\vect \beta}_\text{d}$, when it brings no ambiguity.

\begin{lemma}
For any function $f\!:\!\mathcal{V}  \!\!\to\!\! [0,1]$, there exists a probability measure 
$\vect \beta$ over $\mathcal{C} \!=\! 2^{|\mathcal{V}|}$ s.t. $\pi^\beta_\text{d}(v) = f(v), \forall v \in \mathcal{V}$.
\label{lemma:achievability}
\end{lemma}
\begin{proof}
Let $f$ be an arbitrary function from $\mathcal{V}$ to $[0,1]$. Without loss of generality, we reindex strategies $v$ such that $f$ is non-decreasing, i.e., $\forall v_i, v_j \in \mathcal{V}$, with $i < j$, $f(v_i) \leq f(v_j)$.
Starting with the attack vector $v_1$ with the lowest value of $f$, we assign probability $\beta_{c_1} = f(v_1)$ to the classifier $c_1 \in \mathcal{C}$ with $c_1(v) = 1, \forall v \in \mathcal{V}$. 
We then assign $\beta_{c_2} = f(v_2) -  f(v_1)$ to the classifier $c_2$ with $c_2(v) = 1, \forall v \in \mathcal{V}~\backslash~\{v_1\}$. We continue this process until we reach the last vector $v_{|\mathcal{V}|}$. 
We assign probability $f(v_{|\mathcal{V}|}) -  f(v_{|\mathcal{V}|-1})$
 to the classifier that classifies only $v_{|\mathcal{V}|}$ as coming from an attacker.
The remaining weight $1 -f(v_{|\mathcal{V}|})$ (if positive) is given to the classifier that never classifies the agent as an attacker. 
The strategy $\vect \beta$ derived with the above procedure is guaranteed by construction to have elements in $[0,1]$ with unit sum. 
Moreover $\pi^{\vect \beta}_\text{d}(v_1) = \beta_{c_1} = f(v_1)$ since $v_1$ is classified as coming from an attacker only by $c_1$, 
 $\pi^{\vect \beta}_\text{d}(v_2) = \beta_{c_1}  + \beta_{c_2}= f(v_2)$ since $v_2$ is classified as coming from an attacker only by $c_1$ and $c_2$, and so on until $\pi^{\vect \beta}_\text{d}(v_{|\mathcal{V}|}) =\sum_{i = 1}^{|\mathcal{V}|} \beta_{c_i} = f(v_{|\mathcal{V}|})$ since $v_{|\mathcal{V}|}$ is classified as coming from an attacker by all classifiers $c_1$ through $c_{v_{|\mathcal{V}|}}$.
Thus we have constructed a valid probability measure $\vect \beta$ over $\mathcal{C}$, with $\pi^{\vect \beta}_\text{d}(v) =  f(v), \forall v \in \mathcal{V}$. \qed
\end{proof}

Without loss of generality, we now order the attack vectors in increasing attacker reward, i.e., $R(v_i) \leq R(v_{i+1}),~ \forall i \in \{1, \ldots, |\mathcal{V}|-1\}$. The following lemma, illustrated in Fig.~\ref{fig:regions}, establishes results on the players supports.

\begin{figure}[t]
  \centering \includegraphics[scale = 0.7]{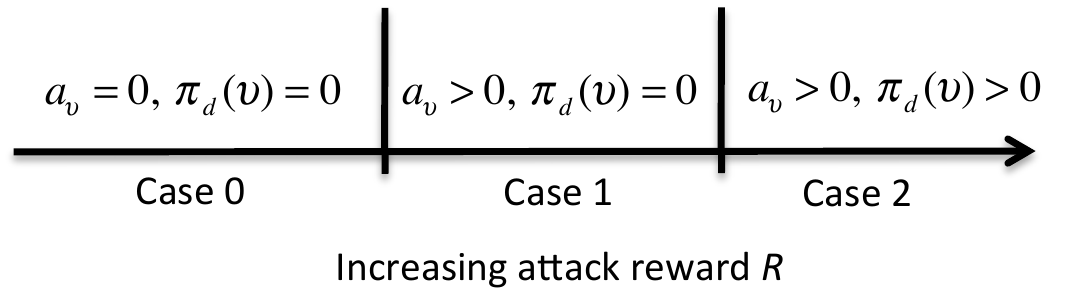}
    \caption{Three cases for strategies selected in NE (see Lemma~\ref{cases_lemma}).}
        \label{fig:regions}
\end{figure}

\begin{lemma}
If $(\vect \alpha, \vect \beta)$ is a NE of $\mathcal{G} = (\mathcal{V}, \mathcal{C}, p, c_\text{d}, c_\text{fa}, P_\text{N}, R(\cdot))$ that yields a probability of detection function $\pi_\text{d}$,  then $\forall v \in \mathcal{V}$ such that $P_\text{N}(v) > 0$, one of the following three cases hold.
\begin{itemize}
\item Case 0: $\alpha_v = 0$ and $\pi_\text{d}(v) = 0$,
\item Case 1: $\alpha_v > 0$ and $\pi_\text{d}(v) = 0$,
\item Case 2: $\alpha_v > 0$ and $\pi_\text{d}(v) > 0$.
\end{itemize}
Furthermore $R(v_0) \leq R(v_1) < R(v_2)$, for any strategies $v_0, v_1, v_2$ in Cases 0, 1, and 2 respectively. 
\label{cases_lemma}
\end{lemma}
\begin{proof}
Suppose that $(\vect \alpha, \vect \beta)$ is a NE and that there exists $v^*\in \mathcal{V}$ with $P_N(v^*)>0$ such that $\alpha_{v^*} = 0$ and $\pi_\text{d}^{\vect \beta}(v^*) >0$. Let $\vect{\hat{\beta}}$ be a mixed strategy of the defender assigning zero detection probability on $v^*$ and leaving the probability of detection unchanged for other attack vectors, i.e., such that $\pi_\text{d}^{\vect{\hat{\beta}}} (v^*) = 0$ and $\pi_\text{d}^{\vect{\hat{\beta}}} (v) = \pi_\text{d}^{\vect \beta} (v)$, for all $v \neq v^*$. By Lemma~\ref{lemma:achievability}, strategy $\vect{\hat{\beta}}$ exists. By Lemma~\ref{same_pd}, we have 
\begin{equation*}
U^\text{D}(\vect \alpha, \vect{\hat{\beta}}) = U^D(\vect \alpha, \vect \beta) + \frac{1-p}{p} c_{\text{fa}} P_\text{N} (v^*) \pi_\text{d}^{\vect{\beta}}(v^*) > U^\text{D}(\vect \alpha, \vect \beta),
\end{equation*}
which contradicts the fact that $(\vect{\alpha}, \vect{\beta})$ is a NE. 

We now show that $R(v_0) \leq R(v_1) < R(v_2), \forall v_0, v_1, v_2$ in Cases 0, 1, and 2 respectively.  Since both pure strategies $v_1, v_2$ are included in the attacker's equilibrium mixed strategy, 
$U^A(v_1, \vect \beta) = U^A(v_2, \vect \beta)$. Since $\pi_\text{d}(v_1) = 0$ and $\pi_\text{d}(v_2) > 0$, 
$$R(v_1) - c_\text{d} \cdot 0 = R(v_2) - c_\text{d} \pi_\text{d}(v_2) \Rightarrow R(v_1) < R(v_2).$$
Moreover, since $\alpha_{v_0} = 0, \alpha_{v_1} > 0$, $U^A(v_0,\vect  \beta) \le U^A(v_1, \vect \beta)$. Since $\pi_\text{d}(v_0) = \pi_\text{d}(v_1) = 0$, this implies $R(v_0) - c_d \cdot 0 \le R(v_1) - c_d \cdot 0$, hence $R(v_0) \le R(v_1)$. \qed
\end{proof}

Lemma~\ref{cases_lemma} shows that, under some assumptions about the non-attacker's distribution, in NE: (1) the defender is never classifying as attacker upon seeing an attack vector that the attacker never uses, and (2) the attacker randomizes only amongst the most rewarding attack vectors and the defender randomizes only amongst classifiers that classify as attacker upon seeing the most rewarding attack vectors. This is illustrated in Section~\ref{sec:sims}, using numerical experiments.

We can also show the following corollary.

\begin{corollary}
If $(\vect \alpha, \vect \beta)$ is a NE of $\mathcal{G}= (\mathcal{V}, \mathcal{C}, p, c_\text{d}, c_\text{fa}, P_\text{N}, R(\cdot))$, for all $v_i, v_j$ in Case 1, $R(v_i) = R(v_j)$.
\end{corollary}
\begin{proof}
Let $v_1, v_2$ be in Case 1. Then $\alpha_{v_1}>0, \alpha_{v_2} > 0$ and 
\begin{align}
\pi_\text{d}(v_1) = \pi_\text{d}(v_2) = 0.\label{eq:zero_prob_det}\end{align}
Since the attacker mixes among both pure strategies $v_1, v_2$, 
these give the same expected utility to the attacker, we have:
\begin{align*}
U^\text{A}&(v_1,\vect \beta) = U^\text{A}(v_2,\vect \beta) \\
&\Rightarrow R(v_1) - c_\text{d} \cdot 0  = R(v_2) -  c_\text{d} \cdot 0, ~~\text{(using \eqref{eq:zero_prob_det})} 
\end{align*}
hence $R(v_1) = R(v_2)$.
\qed
\end{proof}

What we show next is that, under certain assumptions on the non-attacker's behavior, vectors of higher attacker reward have higher or equal probability of getting detected. 

\begin{lemma}
Let $v_1, v_2 \in \mathcal{V}$ be such that $P_N(v_1), P_N(v_2)>0$. In any NE $(\vect \alpha, \vect \beta)$  of $\mathcal{G} = (\mathcal{V}, \mathcal{C} , p, c_\text{d}, c_\text{fa}, P_\text{N}, R(\cdot))$ that yields a probability of detection function $\pi_d$, we have $R(v_1) \leq R(v_2) \Rightarrow \pi_\text{d}(v_1)\leq \pi_\text{d}(v_2)$. 

Hence, if $P_\text{N}(v) >0, \forall v \in \mathcal{V}$, $\pi_\text{d}(v)$ is non-decreasing in the attack reward $R(v)$. 
\label{monotone_pd}
\end{lemma}
\begin{proof}
Let $v_1, v_2 \in \mathcal{V}$ be such that $P_N(v_1), P_N(v_2)>0$ and $R(v_1) \leq R(v_2)$. By Lemma~\ref{cases_lemma}, $v_1, v_2$ are in Case 0, 1, or 2. If both $v_1, v_2$ are in either Case 0 or 1, then $\pi_\text{d}(v_1) = \pi_\text{d}(v_2) = 0$ so that $\pi_\text{d}(v_1)\leq
\pi_\text{d}(v_2)$ holds. Similarly the result holds if $v_1, v_2$ are in Cases 0 or 1 and 2 respectively,  since $\pi_\text{d}(v_1) = 0$ and $\pi_\text{d}(v_2) > 0$.

The only remaining case is when both $v_1$ and $v_2$ are in Case 2. 
By Lemma~\ref{cases_lemma}, we have $\alpha_{v_1} > 0, \alpha_{v_2} > 0$. Since the attacker mixes among both $v_1, v_2$, then $U^A(v_1,\vect \beta) = U^A(v_2,\vect \beta)$. 
But, from \eqref{eq:u_na} we have 
$$U^A({v_2,\vect \beta}) \!=\! U^A({v_1,\vect \beta}) + R(v_2) -  R(v_1) + c_\text{d} \left(\pi_\text{d}(v_1) \!-\! \pi_\text{d}(v_2)\right),$$
so that $R(v_2) \ge R(v_1)$ implies $\pi_\text{d}(v_1) \le \pi_\text{d}(v_2)$.
\qed
\end{proof}

Recall that Lemma~\ref{lemma:achievability} (and its proof) give a way to construct, for any probability of detection function targeted, a defender's strategy $\vect \beta$ that does yield this probability of detection function. The first step in the proof of Lemma~\ref{lemma:achievability} was to reindex attack vectors so that they have non-decreasing detection probability.
By Lemma~\ref{monotone_pd}, vectors ranked in non-decreasing reward already satisfy this property (if $P_\text{N}(v) >0, \forall v \in \mathcal{V}$). We can thus skip the step of reindexing and describe the nature of classifiers $c \in \mathcal{C}$ that are given positive weight. 
Classifier $c_1$ detects all attack vectors and is equivalent to a threshold classifier with threshold equal to the reward of the vector with the smallest reward (or smallest detection probability). Classifier $c_2$ detects all vectors except the one with the smallest reward and is equivalent to a threshold classifier with threshold equal to the second smallest attack reward, and so on until we reach classifier $c_{|\mathcal{V}|}$ that detects only the attack vector with highest reward (threshold equal to the highest reward). The remaining weight (if any) is given to classifier  $c_{|\mathcal{V}|+1}$ that always classifies the agent as a non-attacker, which is a threshold classifier with threshold larger than the highest reward $R(v_{|\mathcal{V}|})$. The above procedure leads to the following corollary of Lemma~\ref{lemma:achievability}:

\begin{corollary}
For any NE
$(\vect \alpha, \vect \beta)$ of $\mathcal{G} = (\mathcal{V}, \mathcal{C}, p, c_\text{d}, c_\text{fa}, P_\text{N}, R(\cdot))$ that results in a non-decreasing probability of detection $\pi^{\vect \beta}_\text{d}$ there exists a NE $(\vect \alpha, \vect{\hat{\beta}})$ of $\mathcal{G}$, where $\vect{\hat{\beta}}_c = 0, \forall c \in \mathcal{C} - \mathcal{C}^T$ and $\pi^{\vect \beta}_\text{d}(v) = \pi^{\vect{\hat{\beta}}}_\text{d}(v), ~\forall v \in \mathcal{V}$. 
\label{corollary:mix_thresholds}
\end{corollary}

We now give the proof of Theorem~\ref{threshold_thm}.
\begin{proof}[Proof of Theorem~\ref{threshold_thm}]
Let $(\vect \alpha, \vect \beta)$ be a NE of $\mathcal{G} = (\mathcal{V}, \mathcal{C}, p, c_\text{d}, c_\text{fa}, P_\text{N}, R(\cdot))$, that results in a probability of detection function $\pi^{\vect \beta}_\text{d}$. 

\underline{Step 1.} First assume that $P_\text{N}(v) > 0$, for all $v \in \mathcal{V}$. By Lemma~\ref{monotone_pd}, $\pi^{\vect \beta}_\text{d}$ is non-decreasing hence we can directly apply Corollary~\ref{corollary:mix_thresholds}, which gives a probability $\vect{\hat{\beta}}$ over $\mathcal{C}$ with positive weight only on $\mathcal{C}^T$ such that $(\vect \alpha, \vect{\hat{\beta}})$ is a NE of $\mathcal{G}$. Therefore $(\vect \alpha, \vect{\beta}^T)$ is also NE of $\mathcal{G}^T = (\mathcal{V}, \mathcal{C}^T, p, c_\text{d}, c_\text{fa}, P_\text{N}, R(\cdot))$, where $\vect \beta^T$ is a probability measure over $\mathcal{C}^T$ with $\beta^T_{c} = \hat{\beta}_{c},$ for all $c\in \mathcal{C}^T$ and  $\pi^{\vect \beta^T}_\text{d}(v) = \pi^{\vect{\hat{\beta}}}_\text{d}(v)= \pi^{\vect{\beta}}_\text{d}(v)$, for all $v \in \mathcal{V}$.

\underline{Step 2.} Second, assume that there exists a unique $v^* \in \mathcal{V}$ with $P_\text{N}(v^*) = 0$. Note that, if $\alpha_{v^*} > 0$, then $\pi_\text{d}(v^*) = 1$. Indeed, suppose $\pi_\text{d}(v^*) < 1$. Consider $\vect{\hat{\beta}}$ that results in $\pi_\text{d}^{\vect{\hat{\beta}}} (v^*) = 1$ and $\pi_\text{d}^{\vect{\hat{\beta}}} (v) = \pi_\text{d}^{\vect \beta} (v)$, for all $v \neq v^*$. By Lemma~\ref{lemma:achievability}, such $\vect{\hat{\beta}}$ exists. By Lemma~\ref{same_pd}, 
\begin{equation*}
U^\text{D}(\vect \alpha, \vect{\hat{\beta}}) = U^\text{D}(\vect \alpha, \vect \beta) + c_\text{d}  \left(1 - \pi_\text{d}^{\vect \beta}(v^*) \right) > U^\text{D}(\vect \alpha, \vect \beta),
\end{equation*}
which contradicts the fact that $(\vect{\alpha}, \vect{\beta})$ is a NE. 
We distinguish two sub-cases: 
\begin{enumerate}[(a)]
\item $v^*$ is not the highest reward vector, i.e., there exists $\hat{v} \in \mathcal{V}$, with $R(\hat{v}) > R({v^*})$. In that case, $\alpha_{v^*} = 0$. Indeed, suppose that $\alpha_{v^*} > 0$. By the previous analysis, $\pi_\text{d}(v^*) = 1$. 
By Lemma~\ref{same_pd} we have
\begin{align*}
 U^A(\hat{v}, \vect \beta) &= U^A({v^*}, \vect \beta) +  R(\hat{v}) - R(v^*) +c_\text{d} \cdot (1 - \pi_\text{d}(\hat{v})),\\
  &> U^A({v^*}, \vect \beta),
\end{align*}
since $R(\hat{v})>R(v^*)$ and $\pi_\text{d}(\hat{v}) \leq 1$. Contradiction.

Let $\vect{\tilde{\beta}}$ be such that $\pi_\text{d}^{\vect{\tilde{\beta}}} (v^*) = \pi_\text{d}^{\vect{{\beta}}} (v^*_\text{next})$ (where $v^*_\text{next}$ is the next rewarding strategy after $v^*$) and $\pi_\text{d}^{\vect{\tilde{\beta}}} (v) = \pi_\text{d}^{\vect \beta} (v)$, for all $v \neq v^*$. By Lemma~\ref{lemma:achievability}, such $\vect{\tilde{\beta}}$ exists. Since $\alpha_{v^*} = 0$, $(\vect \alpha, \vect{\tilde{\beta}})$ is still a NE of $\mathcal{G}$, with the same pair of payoffs as $(\vect \alpha, \vect \beta)$, but with probability of detection function $\pi_\text{d}^{\vect{\tilde{\beta}}}$ non-decreasing in the attack reward. We can then apply Corollary~\ref{corollary:mix_thresholds} and conclude in the same way as in Step 1; and $\pi^{\vect \beta^T}_\text{d}(v)$ and $\pi^{\vect{\beta}}_\text{d}(v)$ will differ only for $v=v^*$.

\item $v^*$ is the highest reward vector, i.e., for all $v \in \mathcal{V}, R(v^*)>R(v)$. In that case, either $\alpha_{v^*}>0$ and then $\pi_\text{d}(v^*) = 1$ so that $\pi_\text{d}$ is non-decreasing in the attack reward and we conclude as in Step 1; or $\alpha_{v^*}=0$ and then we define $\vect{\tilde{\beta}}$ such that $\pi_\text{d}^{\vect{\tilde{\beta}}} (v^*) = 1$ and $\pi_\text{d}^{\vect{\tilde{\beta}}} (v) = \pi_\text{d}^{\vect \beta} (v)$ for all $v \neq v^*$ and conclude as in Step 2(a). 

\end{enumerate}

\underline{Step 3.} Finally, the procedure above generalizes straightforwardly if there exist several attack vectors with $P_N(v)=0$, hence concluding the proof. 
\qed
\end{proof}

\subsection{Reduced attacker's strategy space and equilibrium structure}\label{sec.reduction-attacker}

We now turn to the attacker's equilibrium strategy. We first show Lemma~\ref{scaled_spam}, which gives the attacker's equilibrium strategy for most attack vectors. Then we show that this result, together with the reduction of the defender's strategy space, allow us to reduce the attacker's strategy space as well.
 \begin{lemma}
If $(\vect \alpha, \vect \beta)$ is a NE of   $\mathcal{G} =
(\mathcal{V}, \mathcal{C}, p, c_\text{d}, c_\text{fa}, P_\text{N}, R(\cdot))$, then for all
$v \in \mathcal{V}$ such that $0 < \pi_\text{d}(v) < 1$,
\begin{equation} 
\alpha_{v} =  \frac{1-p}{p} \displaystyle \frac{c_{\text{fa}}}{c_\text{d}}
P_\text{N}(v).
\label{eq:proportional}
\end{equation}
\label{scaled_spam}
\end{lemma}
\begin{proof}
Consider $v_i \in \mathcal{V}$ with $\pi_\text{d}(v_i) \in (0, 1)$. Since $\pi_\text{d}(v_i) \neq 0$, there exists $c_i \in \mathcal{C}$ s.t. $c_i(v_i) = 1$ with $\beta_{c_i} > 0$.
 Since $\pi_\text{d}(v_i) \neq 1$, there exists $c^* \in \mathcal{C}$ s.t. $c^*(v_i) = 0$ with $\beta_{c^*} > 0$.
Now suppose that $c_i(v) = c^*(v), \forall v \in \mathcal{V} - \{ v_i \}$. This is without loss of generality. Indeed, even if we cannot find in the support of $\vect \beta$ two such classifiers, we can construct another defender's strategy $\vect{\tilde{\beta}}$ that contains two such classifiers and has the same probability of detection function. We do that using a construction similar to the one in the proof of Lemma~\ref{lemma:achievability}. If other vectors have the same reward as $v_i$, we include separately each classifier that detects those vectors as attacks. Since $\vect \beta$ and $\vect{\tilde{\beta}}$ have the same probability of detection function, $(\vect \alpha, \vect{\tilde{\beta}})$ is also a NE.

Finally, since $\beta_{c_i}>0, \beta_{c^*} >0$, $U^\text{D}(\vect \alpha, c_i)= U^\text{D}(\vect \alpha, c^*)$, that is
\begin{align*}
& \sum_{v \in \mathcal{V}} \!\alpha_v \!\left( R(v) \!-\! c_d \mathds{1}_ {c_i(v)=1} \right) + \frac{1\!-\!p}{p} c_{\text{fa}} \!\!\sum_{v' \in \mathcal{V}} P_\text{N}(v')\mathds{1}_ {c_i(v')=1} \\
= & \sum_{v \in \mathcal{V}} \!\alpha_v \!\left( R(v) \!-\! c_d \mathds{1}_ {c^*(v)=1} \right) + \frac{1\!-\!p}{p} c_{\text{fa}} \!\!\sum_{v' \in \mathcal{V}} P_\text{N}(v')\mathds{1}_ {c^*(v')=1}.
\end{align*}
This yields
\begin{equation*}
- c_d \alpha_{v_i} + \frac{1-p}{p} c_{\text{fa}} P_\text{N}(v_i) = 0,
\end{equation*}
which immediately gives the result.
\qed
\end{proof}

Lemma~\ref{scaled_spam} shows that for attack vectors $v$ with some uncertainty of getting detected, the attacker mixes proportionally to the non-attacker's distribution. Note that the result clearly also applies to any Nash equilibrium of $\mathcal{G}^T$.
Yet, the attacker's strategy space could be large and complex.  
For instance, consider a game in which there are $M$ different targets and the attacker chooses at each time step from 1 to $N$ a target to attack in $\{ 0, \ldots, M\}$ (counting the no-attack case). An attack $v$ is a sequence of $N$ attacks and $\mathcal{V}$ is the set of all such sequences which is of cardinality ${(M+1)}^N$.
If the defender's classification rule is a threshold on the attack reward, however, permutations of attacks that yield the same attack reward will have the same probability of getting detected. Hence, as we explain below, they can be collapsed into a single strategy.
For instance, if only the total number of times each target is hit (and not the order of the attacks) matters for the reward of the attacker, there are only $N+M+1\choose N$ combinations of attack rewards, and then at most that many unique values of the attacker's reward. Hence, we can exploit the fact that the defender uses only threshold classifiers in $\mathcal{G}^T$ and reduce the cardinality (and complexity) of the attacker's strategy space. As $N$ or $M$ increase, the benefits from this reduction become more profound.

The intuition behind the reduction in the attacker's strategy space comes from the following observation: 
If the attacker includes in his equilibrium support one attack vector of a certain reward, he should include all vectors of the same reward since all of them will have the same probability of detection. By Lemma~\ref{scaled_spam}, the attacker's equilibrium weight on each one is proportional to the non-attacker's distribution. Since the defender's classification is based on the reward of the attack vectors (and not on the actual vector), a game in which the attacker mixes on attack rewards (instead of attack vectors) does not influence the defender's equilibrium strategy. Furthermore, such a game does not give any more or less freedom to the attacker, but reduces the complexity (cardinality) of the strategy space of the attacker.

\begin{figure}[t]
  \centering
    \includegraphics[scale=0.5]{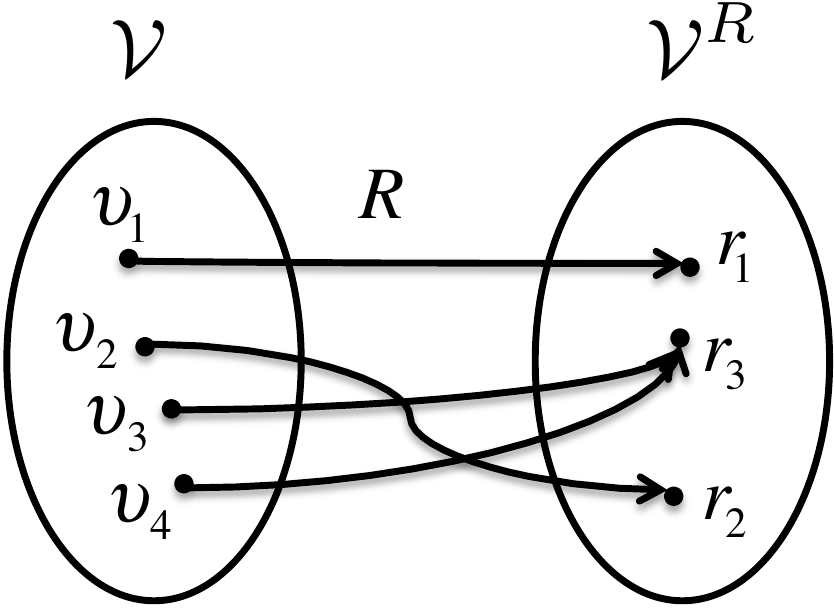}
      \caption{Attacker's initial and reduced strategy spaces.}
    \label{fig:v_r}
\end{figure}

We now formalize the reduction of the game $\mathcal{G}^T$ to a game $\mathcal{G}^{R, T}$. The new attacker's strategy space is defined as the set of all images of the reward function $R: \mathcal{V} \rightarrow  \mathbb{R}_+$ (see an illustration in Fig.~\ref{fig:v_r}):
\begin{definition}[Reduced strategy space of the attacker]
\begin{equation}
\mathcal{V}^R = \{ r \in  \mathbb{R}_+: r = R(v),~\text{for some } v \in \mathcal{V}\}.
\label{eq:v_r}
\end{equation}
\label{reduced_att}
\end{definition}
Note that, although $\mathcal{V}^R$ is not rigorously a subset of $\mathcal{V}$, $\mathcal{V}^R$ is a reduced strategy space in the sense that $R$ is clearly a surjection from $\mathcal{V}$ to $\mathcal{V}^R$.
The non-attacker's probability measure can be similarly reduced to a
probability measure that describes the
non-attacker's distribution on $\mathcal{V}^R$ with
\begin{equation}
P^R_\text{N} (r) = \displaystyle \sum_{v^{\prime}} P_\text{N} (v^{\prime} \in \mathcal{V})  \mathds{1}_ {R(v^{\prime})=r}.
\label{p_nr}
\end{equation}

Finally, since in $\mathcal{G}^T$, $\vect \beta$ is a probability on $\mathcal{C}^T$, we have that $c(v)$ is the same for all $v$ with the same reward and that any two attack vectors with the same reward have the same probability of detection. Hence, by abuse of notation, we can define $c(r)$ as $c(v)$ for any $v$ such that $R(v)=r$, and we can define the probability of detection function as a function of the reward by $\pi_\text{d}(r) := \pi_\text{d}(v)$, where $r = R(v)$. The reduced game $\mathcal{G}^{R, T} = (\mathcal{V}^R, \mathcal{C}^T, p, c_\text{d}, c_\text{fa}, P^R_\text{N})$ is then defined as the game between the attacker choosing $r\in \mathcal{V}^R$ and the defender choosing $c\in \mathcal{C}^T$ where the utilities adapt \eqref{eq.UA}-\eqref{eq.UDhat} in the obvious way: 
\begin{align*}
U^\text{A}(r, c) & = r - c_\text{d}  \mathds{1}_ {c(r)=1},\\
U^\text{D}(r, c) & = - U^\text{A}(r,c) - \displaystyle \frac{1-p}{p} c_\text{fa} \sum_{r' \in \mathcal{V}^R} P^R_N(r') \mathds{1}_ {c(r')=1}.
\end{align*}
The expected utilities in mixed strategies also adapt \eqref{eq:u_na}-\eqref{eq:u_nd} in the obvious way.

The following proposition formalizes the relationship between the NE of $\mathcal{G}^T$ and $\mathcal{G}^{R, T}$.
\begin{proposition}
If $(\vect \alpha, \vect \beta)$ is a NE of $\mathcal{G}^T =
(\mathcal{V}, \mathcal{C}^T, p, c_\text{d}, c_\text{fa}, P_\text{N}, R(\cdot))$, then $(\vect \alpha^*, \vect \beta)$ is a NE of $\mathcal{G}^{R, T} =(\mathcal{V}^R, \mathcal{C}^T, p, c_\text{d}, c_\text{fa}, P^R_\text{N})$ with the same equilibrium payoff pair
where $\alpha^*_{r_i} = \displaystyle \sum_{v_j \in \mathcal{V}, R(v_j) = r_i} \alpha_{v_j}, ~\forall r_i \in \mathcal{V}^R.$
\label{attacker_reduced_prop}
\end{proposition}

The proof of Proposition~\ref{attacker_reduced_prop} can be found \techrep{in the Appendix}{in our technical report \cite{techreport}}. 
The importance of Proposition~\ref{attacker_reduced_prop} comes from the fact that it is easier to compute the NE $(\vect \alpha^*, \vect \beta)$ of $\mathcal{G}^{R, T}$, in which the cost matrix of the attacker consists of non-identical rows. This NE is given in Theorem~\ref{aggregate_theorem} in Section~\ref{sec:analysis}. However, from the NE $(\vect \alpha^*, \vect \beta)$ of $\mathcal{G}^{R, T}$, we can easily recover a NE $(\vect \alpha, \vect \beta)$ of $\mathcal{G}^{T}$ as follows. 
By Proposition~\ref{attacker_reduced_prop}, $\vect \beta$ is unchanged.
Given $\vect \alpha^*$ on $\mathcal{V}^R$ we compute the attacker's strategy $\vect \alpha$ over $\mathcal{V}$, as follows:
For $r_i \in \mathcal{V}^R$ with $\pi_\text{d}(r_i) \in (0,1)$, $\alpha_{v_i}$ is given by \eqref{eq:proportional} $\forall v_i \in \mathcal{V}$ with $R(v_i) = r_i$, by Lemma~\ref{scaled_spam}. For $r_i$ with $\pi_\text{d}(r_i) \in \{ 0, 1\}$, any possible combination of weights is possible, as long as  $\sum_{v_j \in \mathcal{V}, R(v_j) = r_i} \alpha_{v_j} = \alpha^*_{r_i}$.
Hence, although we reduce the attacker's strategy space, we provide a roadmap to get all NE of $\mathcal{G}^T$ as well. Note that if all attack strategies in $\mathcal{V}$ yield distinct attack reward, then $|\mathcal{V}| = |\mathcal{V}^R|$ and there is no reduction in the attacker's strategy space.

\section{Nash Equilibrium Analysis}
\label{sec:analysis}

Our goal in this section is to characterize the structure of the NE of the classification game.
It is known that every finite game (finite number of players with finite number of strategies for each player) has a mixed-strategy NE~\cite{Nash50}.
Our game is finite, thus it admits a NE in mixed strategies. However, finding the Nash equilibrium has a high computational complexity in the general case \cite{Chen:2009:SCC:1516512.1516516}.

We consider the game  $\mathcal{G}^{R,T}= (\mathcal{V}^R, \mathcal{C}^T, p, c_\text{d}, c_\text{fa}, P^R_\text{N})$, in which the attacker's strategy space consists of distinct attack rewards $r \in \mathcal{V}^R$, the defender's strategy space consists of threshold classifiers $c \in \mathcal{C}^T$, and $P^R_\text{N}$ is the non-attacker's probability measure on $\mathcal{V}^R$ given by \eqref{p_nr}.
We denote by $r_i$, $i \in \{1, \ldots, |\mathcal{V}^R|\}$ the elements of $\mathcal{V}^R$ and, without loss of generality, we assume that they are ranked in increasing order, i.e., $r_{i}<r_{i+1}$ for all $i\in \{ 1, \ldots, |\mathcal{V}^R|-1\}$. 
Similarly, classifier $c_i$ corresponds to a threshold classifier with threshold equal to the attacker reward $r_i$. Recall (see Definition~\ref{def.CT}) that a threshold classifier with threshold $t$ classifies as attacker if the reward is $r\ge t$. Hence, $c_1$ corresponds to the ``always classify as attacker" strategy. By definition, $\mathcal{C}^T$ also includes the ``always classify as non-attacker" strategy, which is denoted by $c_{|\mathcal{V}^R|+1}$ and corresponds to a threshold $r_{|\mathcal{V}^R|} + \delta$ for any $\delta>0$. Hence we have $| \mathcal{C}^T| =  |\mathcal{V}^R| + 1$.

We can express the payoff functions of the players in compact form as matrices. We define $\tilde{\Lambda}$ to be the cost matrix of the attacker, with $\tilde{\Lambda}(i,j) = c_d \mathds{1}_{r_i \geq  r_j} - r_i$, $i \in \{1, \ldots, |\mathcal{V}^R|\}$, $j \in \{1, \ldots, |\mathcal{V}^R + 1|\}$:
\begin{equation*}
  \tilde{\Lambda} = c_\text{d} \begin{pmatrix} 1 & 0 & \cdot \cdot \cdot &
    \cdot \cdot \cdot & 0&0\\ \vdots & 1 & \ddots & & \vdots &\vdots
    \\ \vdots & & \ddots & \ddots & \vdots & \vdots \\ \vdots & &&
    \ddots & 0 &\vdots \\ 1 &\cdots &\cdots& \ldots & 1 & 0 \\
 \end{pmatrix}
 - \begin{pmatrix} r_1 \\ r_2 \\ r_3 \\ \vdots
   \\ r_{{|\mathcal{V}^R|-1}} \\ r_{{|\mathcal{V}^R|}}\\
 \end{pmatrix} \cdot \mathbf{1}'_{|\mathcal{V}^R| + 1}.
 \label{eq:lambda}
 \end{equation*}
Recall that we use the \textit{prime} sign (${}^{\prime}$) for transpose of matrices and vectors. There are $|\mathcal{V}^R|$ rows (strategies) for the attacker, ranked by increasing reward. 
Certain computations are simplified by using a matrix with only positive entries. Therefore, we define $\Lambda \triangleq \tilde{\Lambda} + (r_{{|\mathcal{V}^R|}}+ \epsilon)\cdot 1_{{|\mathcal{V}^R|} \times{|\mathcal{V}^R| + 1}}, \text{ where }
\epsilon >0$.
Adding a constant to the players' payoff does not affect their best responses, and hence does not change the equilibrium strategies. Thus, from here on, we will use matrix $\Lambda$ to define the players' payoffs.

In the following, the pair $(\vect \alpha, \vect \beta)$ will denote the probability measures of the attacker and defender on $\mathcal{V}^R, \mathcal{C}^T$ respectively. The attacker's expected cost is given by $\vect \alpha^{\prime} \Lambda \vect \beta$.
The defender's expected payoff is given by $\vect \alpha^{\prime} \Lambda^\text{eq} \vect
\beta$, with  \begin{equation}
\Lambda^\text{eq} = \Lambda - \vect 1_{|\mathcal{V}^R|} \cdot \vect
\mu^{\prime},
 \label{eq:def_payoff}
 \end{equation} where $\vect \mu$  represents
the false alarm penalty vector for the defender with elements $\mu_i =  \displaystyle \frac{1-p}{p}
c_\text{fa} \sum_{k \geq i} P^R_\text{N}(r_k)$.
 We assume that $\vect \mu$ is a strictly decreasing vector (component-wise): $\mu_i > \mu_{i+1}, \forall i \in \{1, \ldots, |\mathcal{V}^R|\}$. This assumption is equivalent to assuming that the non-attacker gives some positive weight to all strategies $r \in \mathcal{V}^R$, i.e.,  $P^R_\text{N}(r_i) > 0,  \forall i \in \{1, \ldots, |\mathcal{V}^R|\}$. Even if this property does not hold, we can still describe how both players behave, as shown in Theorem~\ref{threshold_thm} in Section~\ref{sec:thresholds}.
  
It is easy to show that $\mathcal{G}^{R,T} = (\mathcal{V}^R, \mathcal{C}^T, p, c_\text{d}, c_\text{fa}, P^R_\text{N})$ is best-response equivalent (see Definition~\ref{equivalence}) to a zero-sum game, in which the attacker's cost is given by $\vect \alpha^{\prime} \Lambda^\text{eq} \vect \beta$. Indeed, the two games have the same players with the same strategy spaces. Vector $\vect \mu$ depends only on the non-attacker's distribution and is fixed.  Adding constants to the columns of the cost matrix of the attacker (row player) in the original classification game yields the cost matrix of the attacker of the new game without changing the Nash equilibria of the game. Indeed, the defender's payoff matrix is unchanged, and, for any given $\vect \beta$, minimizing $\vect {\alpha^{\prime}} \Lambda \vect \beta$ and minimizing $\vect{ \alpha^{\prime}} \Lambda^\text{eq} \vect \beta$ with $\Lambda^\text{eq}(i,j) = \Lambda(i,j) - \mu_j$ give the same minimizing strategy for the attacker.  Thus the two games generate the same sets of best response functions and have the same equilibrium strategies. 
 
Note that the best-response equivalence of our game to a zero-sum game guarantees that in all NE the defender's expected payoff will be the same (and equal to the value of the zero-sum game), but the attacker's payoff is not always the same in each equilibrium of the original nonzero-sum game.
Indeed, equilibria of our original game could give different payoffs to the attacker after transforming his cost matrix $\Lambda^\text{eq}$ (adding constants to the columns) back to $\Lambda$. 
 
There exist polynomial algorithms to compute the Nash equilibria in zero-sum games via a transformation to a linear program (LP)~\cite{Karmarkar:1984:NPA:800057.808695}. These algorithms, however, do not provide structure
on the equilibrium strategies. In the remaining of this section, we aim to give
more intuition on the players' NE strategies than the solution derived
via a linear programming toolbox. Along the way, we also provide an algorithm to compute the NE, which can be faster
than solving the LP. 

Our results can be summarized in the following theorem. 
\begin{theorem}
Algorithm~\ref{algo.main} finds all NE of the classification game $\mathcal{G}^{R,T}$. Moreover, if $(\vect \alpha, \vect \beta)$ is a NE, then, there exists $k\in \{1, \ldots, |\mathcal{V}^R|\}$ such that 
\begin{align*}
\vect \beta  &= (0, \ldots, 0,  \beta_k, \ldots, \beta_{|\mathcal{V}^R|}, \beta_{|\mathcal{V}^R|+1}),\\
\vect \alpha  &= (0, \ldots, 0,  \alpha_k, \ldots, \alpha_{|\mathcal{V}^R|}),
\end{align*}
where 
\begin{align}
\beta_i & =  \displaystyle \frac{r_i - r_{i-1}}{c_\text{d}}, && \forall i \in \{k+1, \ldots, {|\mathcal{V}^R|} \}, \\ 
\label{eq:att_weight}  \alpha_i & = \frac{1-p}{p} \frac{c_{\text{fa}}}{c_\text{d}} P_N^R(r_i), && \forall i \in \{k+1, \ldots, {|\mathcal{V}^R|}-1 \},
\end{align}
and $\beta_k, \beta_{|\mathcal{V}^R|+1} \ge 0$ and $\alpha_k, \alpha_{|\mathcal{V}^R|} \ge 0$ are such that 
\begin{itemize}
\item[($i$)] $ \beta_k \in (0,  \frac{r_k - r_{k-1}}{c_\text{d}}), \beta_{|\mathcal{V}^R|+1} =0$, and\\ 
$\alpha_k$ satisfies  \eqref{eq:att_weight}, $\alpha_{|\mathcal{V}^R|} >0$; or

\item[($ii$)] $ \beta_k = 0, \beta_{|\mathcal{V}^R|+1} > 0$, and\\ 
$\alpha_k \in (0, \frac{1-p}{p} \frac{c_{\text{fa}}}{c_\text{d}} P_N^R(r_k))$, $\alpha_{|\mathcal{V}^R|}$ satisfies  \eqref{eq:att_weight}; or

\item[($iii$)] $ \beta_k = 0, \beta_{|\mathcal{V}^R|+1} = 0$, and\\ 
$\alpha_k \!\! \in \!\! [0, \min (\frac{1-p}{p} \frac{c_{\text{fa}}}{c_\text{d}} P_N^R(r_k),  1\!\!-\!\!\sum_{i = k+1}^{|\mathcal{V}^R|-1} \! \alpha_i)]$, $\alpha_{|\mathcal{V}^R|} \!\!\ge\!\! 0$; or

\item[($iv$)] $ \beta_k \in [0,  \frac{r_k - r_{k-1}}{c_\text{d}}], \beta_{|\mathcal{V}^R|+1} \ge 0$, and\\ 
$\alpha_k, \alpha_{|\mathcal{V}^R|}$ satisfy \eqref{eq:att_weight}.

\end{itemize}

\label{aggregate_theorem}
\end{theorem}

The rest of this section is dedicated to proving Theorem~\ref{aggregate_theorem}.
As a direct consequence of best-response equivalence of our game to a zero-sum game, we have the following result.
\begin{lemma}
In NE, the defender's strategy $\vect \beta$ solves the following linear program (LP):
\begin{equation}
  \begin{aligned}
    & \underset{\vect \beta, z}{\text{maximize}} & & - \vect \mu^{\prime}
    \vect \beta + z\\ & \text{subject to} & & z \mathbf{1}_{|\mathcal{V}^R|} \leq
    \Lambda \vect \beta\\ & & & \mathbf{1}^{\prime}_{|\mathcal{V}^R|+1} \cdot \vect{\beta} =
    1, ~\vect \beta \geq \vect 0.\\
  \end{aligned}
  \label{def_LP}
\end{equation}
\label{lemma_defLP}
\end{lemma}

\begin{proof}
If $(\vect \alpha, \vect \beta)$ is a NE of  $\mathcal{G}^{R,T}$ with attacker's cost matrix $\Lambda$, then $(\vect \alpha, \vect \beta)$ is a NE of the zero-sum, best-response equivalent game with cost matrix $\Lambda^{\text{eq}}$.
Therefore, $\vect \beta$ maximizes $\min_{\vect \alpha} \vect \alpha^{\prime} \Lambda^{\text{eq}} \vect \beta = \min [\Lambda_\text{} \vect \beta] - \vect \mu^{\prime}\vect \beta$. Transforming this optimization problem to an LP we get the program~\eqref{def_LP}.
\qed
\end{proof}

An important consequence of Lemma~\ref{lemma_defLP}, is that the defender's strategy $\vect \beta$ in Nash equilibrium maximizes her minimum payoff. Thus playing $\beta$ gives the defender the robustness property that her expected payoff will not be any worse than her expected Nash equilibrium payoff regardless of what the attacker chooses to play.

We now define the main entities used throughout the section.
\begin{definition}
A {\bf polyhedron} is the solution set of a finite number of linear inequality constraints.
An inequality constraint is {\bf tight} if it holds as an equality; otherwise, it is {\bf loose}. A point $\vect x = (x_1, \ldots, x_{|\mathcal{V}^R|+1})$ of a polyhedron is said to be {\bf extreme} if there is no $\vect {x}'$ whose set of tight constraints is a strict superset of the set of tight constraints of $\vect {x}$. 
For an $n$-dimensional linear program, a point is called a {\bf basic solution}, if $n$ linearly independent constraints are tight for that point.  A {\bf feasible solution} to a linear program is a solution that satisfies all constraints. A point is a {\bf basic feasible solution}, iff it is a basic solution that is also feasible. 
Two distinct basic feasible solutions to an $n$-dimensional linear program are {\bf adjacent} if we can find $n-1$ linear independent constraints that are tight at both of them.
 We say that a point $\vect x$ of a polyhedron {\bf corresponds} to strategy $\vect \beta$ (or strategy $\vect \beta$ corresponds to $\vect x$), if $\vect \beta = \vect x/ \| \vect x\|$.
\label{def_mutliple}
\end{definition}
Extreme point and basic feasible solution are equivalent terms~\cite[Chapter 2.5]{luenberger2003linear} and we will use them interchangeably.

\subsection{Form of optimal extreme points}

In this section, we gain intuition on the structure of the defender's NE strategy $\vect \beta$.  
We first show the following lemma.
\begin{lemma}
Any NE strategy $\vect{\beta}$ of the defender corresponds to an extreme point or a convex combination of extreme points of the polyhedron defined by 
\begin{equation}
P: ~\Lambda_\text{} \vect x \geq \mathbf{1}_{|\mathcal{V}^R|}, \vect x_{|\mathcal{V}^R|+1} \geq \vect 0.\label{eq:lp}
\end{equation}

\label{def_lemma}
\end{lemma}

The proof of Lemma~\ref{def_lemma} can be found \techrep{in the Appendix}{in our technical report \cite{techreport}}. We call the first type of constraints  ``inequality constraints'' and the second type ``positivity constraints.'' There are $|\mathcal{V}^R|$ inequality constraints and $|\mathcal{V}^R|+1$ positivity constraints.
Writing down the inequality constraints, we get
\begin{align*}
  c_\text{d} \cdot  x_1 &+ \left( r_{|\mathcal{V}^R|} - r_1+ \epsilon\right) \cdot \|\vect x\|  &\geq  1 &\\ 
  c_\text{d} \cdot (x_1+x_2) &+\left(r_{|\mathcal{V}^R|}-r_2 + \epsilon\right) \cdot \|\vect x\| & \geq 1&\\ 
& ~\vdots \\
  c_\text{d} \cdot (x_1+x_2+\ldots+x_{|\mathcal{V}^R|}) & +~~~~~~~~~~ \epsilon~~~~~~~~ \cdot \|\vect x\| & \geq 1&.
\end{align*}

Searching for an extreme point of the polyhedron $P$ defined in~\eqref{eq:lp} is computationally straightforward and there are known algorithms that provide polynomial complexity. Our goal is to provide an algorithm that can run faster (though still polynomially) and in parallel to provide intuition on the structure of the extreme points. The main method used is to reduce the search space by eliminating suboptimal non-extreme points.

Combining properties of basic feasible solutions of an LP and of the structure of the defender's LP, we show the following lemma describing the set of tight inequality constraints. 
\begin{lemma}
  At an extreme point $\vect x$ that corresponds to a defender's NE strategy $\vect \beta$, there exists exactly one contiguous block (of indices) of inequality constraints that are tight and the last inequality constraint is tight.
  \label{block_lemma}
\end{lemma}

The proof of Lemma~\ref{block_lemma} can be found \techrep{in the Appendix}{in our technical report \cite{techreport}}. We define $s$ as the index of the first (starting) tight inequality constraints. The lemma states that all inequality constraints from $s$ to $|\mathcal{V}^R|$ are tight.
We can now state the result that describes the form of optimal extreme points of the defender's LP. Its proof is \techrep{in the Appendix}{in our technical report \cite{techreport}}.
\begin{lemma}
Any extreme point $\vect x$ of polyhedron $P$ that corresponds to a defender's NE strategies is of one of the following types: 
\begin{align*}
\text{Type I:}~\vect x &\!=\!\left(0, \ldots, 0, x_{s_1} \! \ge \! 0, x_{s_1+1}\! >\! 0, \ldots, x_{|\mathcal{V}^R|}\! >\! 0,0 \right)^{\prime}\! , \\
\text{Type II:}~\vect x& \!=\!\left(0, \ldots,0, x_{s_2+1}\! >\! 0,\ldots, x_{|\mathcal{V}^R|}\! >\! 0 , x_{|\mathcal{V}^R| + 1} \! \ge \! 0\right)^{\prime}\! , 
    \end{align*}
    for some $s_1, s_2 \in \{1, \ldots, |\mathcal{V}^R| \}$, with $x_i = \displaystyle \frac{r_i-r_{i-1}}{c_\text{d}} \| \vect x\|$ for all $i \in \{ s_j+1, \ldots, |\mathcal{V}^R|\}$ ($j\in\{1, 2\}$).
Moreover, there exist at most one extreme point of type I and two adjacent extreme points of type II that correspond to a defender's NE strategies.
\label{forms_lemma}
\end{lemma}

\subsection{Form of players' equilibrium strategy}

With Lemma~\ref{forms_lemma} that describes the form of the possible extreme points of polyhedron $P$ defining the constraints of the defender's LP, we can now prove our main result, Theorem~\ref{aggregate_theorem}. The complete proof can be found \techrep{in the Appendix}{in our technical report \cite{techreport}}. It essentially works by enumerating all of the possible combinations of optimal basic feasible solutions allowed by Lemma~\ref{forms_lemma} to get the possible forms of $\vect \beta$ (i.e., all possible solutions of the defender's LP which we consider the primal) and using the complementary slackness condition from Linear Programming to get the possible forms of $\vect \alpha$ (i.e., all possible solutions of the dual LP). 

Theorem~\ref{aggregate_theorem} provides both an algorithm (Algorithm~\ref{algo.main}) that finds all NE and a compact characterization of the restricted number of possible forms that a NE can have. Interestingly, we observe that the defender assigns a weight to a reward $r_i$ that is positive and proportional to the marginal reward increase at that point, on a support that goes until the highest reward $r_{|\mathcal{V}^R|}$. This is somewhat counter-intuitive as it implies that the defender includes at NE with positive weights classifiers that almost never classify as attacker even for a high reward and even if the probability that a non-attacker uses this reward is arbitrarily small. We also recover in Theorem~\ref{aggregate_theorem} the fact that the attackers mimics the non-attacker's distribution (proportionally) on a support that corresponds to the defender's support. 

For readers familiar with Linear Programming, let us finally remark that the result in Theorem 2 (and its proof) is in accordance with the relationship between degeneracy and multiplicity of the primal and the dual optimal solutions~\cite[p. 144]{Sierksma02a} (a degenerate optimal solution is a solution of size $n$ where more than $n$ constraints are tight).
If there exists a unique non-degenerate optimal solution of the primal (defender's LP) then the optimal solution to the dual problem (the attacker's LP) is also unique and non-degenerate (cases ($i$) and ($ii$) of Theorem~\ref{aggregate_theorem}).
If the primal has multiple solutions with at least one non-degenerate then the optimal solution to the dual is unique and degenerate (case $(iv)$ of Theorem~\ref{aggregate_theorem}). 
If the primal optimal solution is unique and degenerate, then the dual has multiple optimal solutions (case $(iv)$ of Theorem~\ref{aggregate_theorem}). 
Another benefit of having the solution of the game corresponding to the solution of an LP is that there are well studied methods for analyzing the sensitivity of an LP to parameters (see for instance~\cite{opac-b1079784}). These methods can be used to study how sensitive the game solution is to parameter perturbations. 
\begin{algorithm}
 \For {$type=1,2$}{
construct $\vect \beta$ for $s \in \{1, \ldots, |\mathcal{V}^R|\}$ using Algorithm 2\\
find $(\vect \beta_{1,2}, s^*_{1,2})$ that maximize defender's payoff $U^D_{1,2}$}

 \If{$U^D_1 >U^D_2$}{
 $\vect \beta \leftarrow \text{compute-}\beta (s^*_1,  1)$; $\vect \alpha \leftarrow \text{compute-}\alpha (s^*_1)$
}
 \If{$U^D_1 < U^D_2$}{
 \If{$s_{2}^*$ is unique}{
 $\vect \beta \leftarrow \text{compute-}\beta (s^*_2, 2)$; $\vect \alpha \leftarrow \text{compute-}\alpha (s^*_{2})$
}
 \Else{  
 \small \tt // denote $s^*_{2a}$ and $s^*_{2b} \! = \! s^*_{2a} \! +\! 1 $ the 2 solutions\\[0.5mm]
 $\vect \beta_{a} \!\leftarrow\! \text{compute-}\beta (s^*_{2a},2)$; $\vect \beta_{b} \!\leftarrow\! \text{compute-}\beta(s^*_{2b},2)$\\
 $\vect \beta$ $\leftarrow$ convex hull of $\vect \beta_{a}, \vect \beta_{b}$\\
 $\vect \alpha \leftarrow \text{compute-}\alpha (s^*_{2b})$

}
}
\If{$U^D_1 = U^D_2$}{

\If{$s_{2}^*$ is unique}{
 \small \tt // $s^*_2 = s^*_1$\\[0.5mm]
 $\vect \beta_{1} \leftarrow \text{compute-}\beta(s^*_{1},1)$; $\vect \beta_{2} \leftarrow \text{compute-}\beta (s^*_{1},2)$\\
 \If{$\vect \beta_1 \neq \vect \beta_2$}{
$\vect \beta$ $\leftarrow$ convex hull of $\vect \beta_{1}, \vect \beta_{2}$
     }
   \Else { 
   $\vect \beta  \leftarrow \vect \beta_1$
   }
 $\vect \alpha \leftarrow \text{compute-}\alpha (s^*_1)$
}
\Else{
 \small \tt // denote $s^*_{2a}$ and $s^*_{2b} \! = \! s^*_{2a} \! +\! 1 $ the 2 solutions\\[0.5mm]
 \small \tt // the type 1 and type 2a solutions are identical\\[0.5mm]
 $\vect \beta_{a} \! \leftarrow \! \text{compute-}\beta(s^*_{2a},1)$; $\vect \beta_{b} \! \leftarrow \! \text{compute-}\beta (s^*_{2b},2)$\\
 $\vect \beta$ $\leftarrow$ convex hull of $\vect \beta_{a}, \vect \beta_{b}$\\
 $\vect \alpha \leftarrow \text{compute-}\alpha (s^*_{2b})$

 }

}
\caption{How to compute the NE $(\vect \alpha, \vect \beta$)}
\label{algo.main}
\end{algorithm}

\begin{algorithm}
\SetKwData{Left}{left}\SetKwData{This}{this}\SetKwData{Up}{up}
\SetKwFunction{Union}{Union}\SetKwFunction{FindCompress}{FindCompress}

\BlankLine
\For{$i =  1$ \KwTo $s-1$}{
$\beta_i \leftarrow 0$
}
\For{$i =  s+1$ \KwTo $|\mathcal{V}^R|$}{
$\beta_i \leftarrow \displaystyle \frac{r_{i} - r_{i-1}}{c_\text{d}}$
}
$remainder \leftarrow  \mathbf{1}_{type =1}  (1 - \sum_{s+1}^{|\mathcal{V}^R|} \beta_i)$\\
\If{type=1} {$\beta_s \leftarrow remainder$, $\beta_{|\mathcal{V}^R|+1} \leftarrow 0 $}
\If{type=2} {$\beta_s \leftarrow 0$, $\beta_{|\mathcal{V}^R|+1} \leftarrow remainder$ }
// \tt $U^D(\vect \beta)$ is the defender's NE payoff 
$U^D_{type} \leftarrow  \min[\Lambda \vect \beta] - \vect \mu' \vect \beta$
\caption{Compute-$\beta(s,type)$}\label{algo_beta}
\end{algorithm}

\begin{algorithm}
// \tt $\vect \beta$ is the \emph{set} of all convex combinations if multiple\\
\For{$i =  1$ \KwTo $k-1$}{
$\alpha_i \leftarrow 0$
}
\For{$i =  k+1$ \KwTo $|\mathcal{V}^R|-1$}{
$\alpha_i \leftarrow \displaystyle  \frac{1-p}{p} \frac{c_\text{fa}}{c_\text{d}} P^R_N(r_i)$
}

\If{$\beta_k > 0$ for some $\vect \beta$ }{
$\alpha_k \leftarrow \displaystyle  \frac{1-p}{p} \frac{c_\text{fa}}{c_\text{d}} P^R_N(r_k)$ \\
$\alpha_{|\mathcal{V}^R|} \leftarrow  1-\sum_{i = k}^{|\mathcal{V}^R|-1}  \alpha_i$
}

\ElseIf{$\beta_{|\mathcal{V}^R|+1} > 0$ for some $\vect \beta$}{
$\alpha_{|\mathcal{V}^R|} \leftarrow \displaystyle  \frac{1-p}{p} \frac{c_\text{da}}{c_\text{d}} P^R_N(r_{|\mathcal{V}^R|})$\\
$\alpha_k \leftarrow  1 - \sum_{i = k+1}^{|\mathcal{V}^R|} \alpha_i $\;
}
\ElseIf {$\beta_k = \beta_{|\mathcal{V}^R|+1} = 0$ }{
\For{$i =  1$ \KwTo $k-1$ and $i =  k+1$ \KwTo $|\mathcal{V}^R|-1$}{
$\alpha^1_i \leftarrow \alpha_i$, $\alpha^2_i \leftarrow \alpha_i$
}
$\alpha^1_k \leftarrow 0$\\
$ \alpha^2_k \leftarrow \min \left(\frac{1\!-\!p}{p} \frac{c_{\text{fa}}}{c_\text{d}} P_N^R(r_k),  1 -\sum_{i = k+1}^{|\mathcal{V}^R|-1}  \alpha_i \right)$\\
$\alpha^{1, 2}_{|\mathcal{V}^R|} \leftarrow  1-\sum_{i = k}^{|\mathcal{V}^R|-1}  \alpha^{1, 2}_i$\\
$\vect \alpha$ $\leftarrow$ convex hull of $\vect \alpha^1, \vect \alpha^2$
}

\caption{Compute-$\alpha(\vect \beta, k)$}
\label{algo_alpha}
\end{algorithm}

\section{Numerical Results}
\label{sec:sims}
In this section, we numerically study several instances of our model and observe the behavior of the players in NE.  
In particular, we explore the  strategic attacker's exploitation of the knowledge he has of the non-attacker's distribution (noise).

\subsection{Single-feature-based classification}
We first consider $\mathcal{G}^{R,T} = (\mathcal{V}^R, \mathcal{C}^T, p, c_\text{d}, c_\text{fa}, P^\text{N})$, in which classification is based on a single feature, i.e.,  there is a single target of interest and the defender observes how often this target is (or attempted to be) compromised. The attacker's strategy space consists of $N+1$ attack rewards $r_0, \ldots, r_N$, with $r_i =  i\cdot c_\text{a}$ representing the attack reward when the target is compromised $i$ times. 
The defender's strategy space $\mathcal{C}^T$ consists of all threshold classifiers on $r_i$. 
The attacker bears a cost of $c_\text{d}$ upon detection. 
The false alarm penalty for the defender
when she mistakenly classifies the non-attacker as an attacker is
expressed by the factor $c_\text{fa}$.

The non-attacker, typically a normal user (or
noise from the point of view of defender looking to do attacker detection), accesses the target $i$ times with binomial distribution: attack reward $r_i$ is the outcome of $i$ successes over $N$ trials with probability of success $\theta_0$. His behavior results in a distribution
\begin{equation} 
P^R_\text{N} = {N\choose k} \theta^k_0 (1-\theta_0)^{N-k} 
\label{eq:noise}.
\end{equation}  

\begin{figure}[h!]
  \centering
  \includegraphics[width=0.5\textwidth]{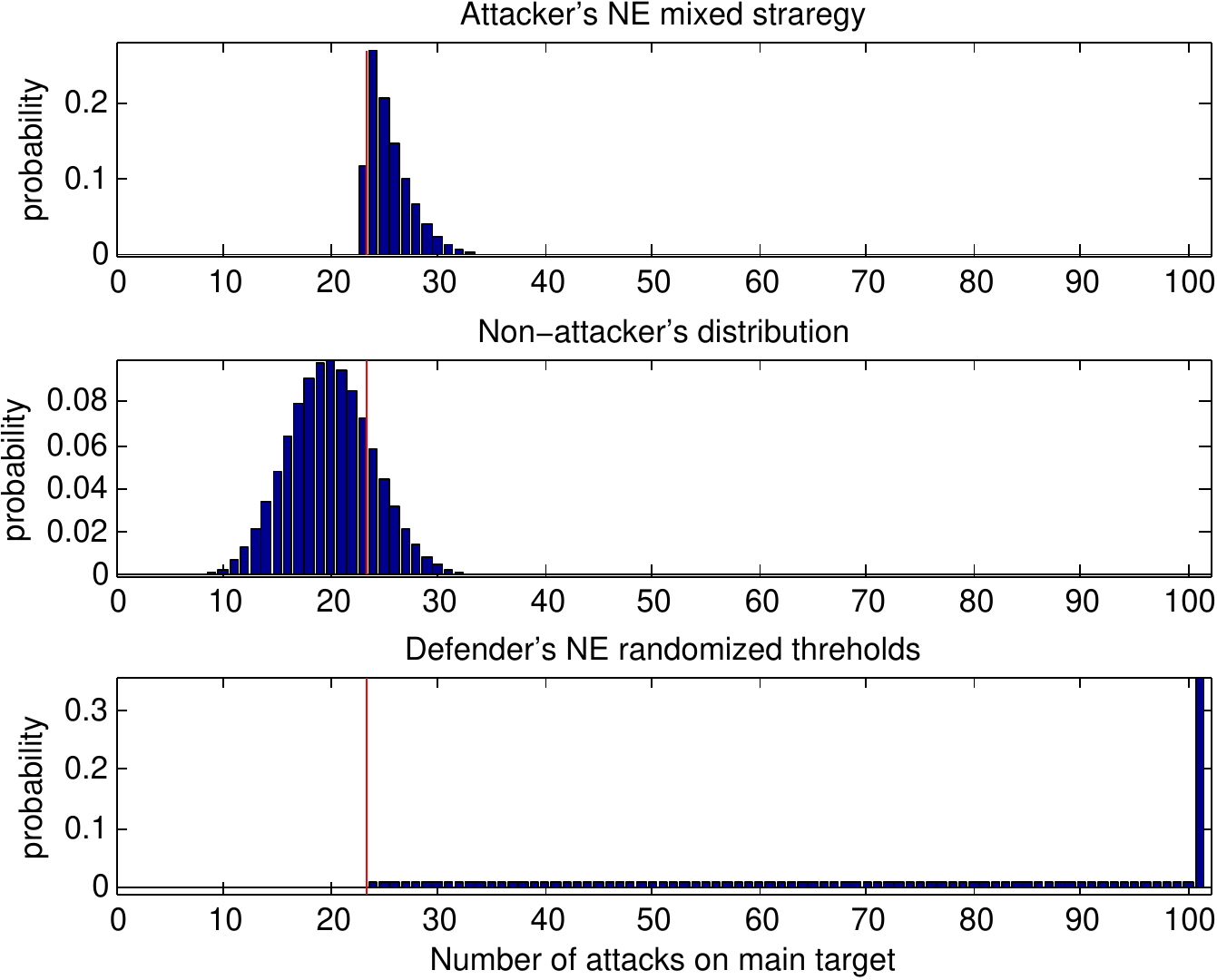}
      \caption{Equilibrium distributions.}
          \label{fig:truncated}
\end{figure}

Figure~\ref{fig:truncated} illustrates the equilibrium mixing distributions of the players, and the probability distribution of non-attackers, for the following particular choice of parameters: $c_\text{a} = 1$, $c_\text{d} = 120$, $p = 0.2$, 
$c_\text{fa} = 140$, $\mathcal{V}^R = \{r_0, \ldots r_i, \ldots,  r_{100}\}$, with $r_i =ic_\text{a}$, and $P^R_\text{N}(r_i)$ given by \eqref{eq:noise} with 
$\theta_0 = 0.2$. 
This example turns out to illustrate all of the major structural findings of the players' equilibrium distributions proved in Theorem~\ref{aggregate_theorem}. In particular: 
\begin{enumerate}
\item  The attacker uses a truncated, scaled version of
the distribution of the non-attacker, but only on
a subset of the support (the interior of the defender's support). Moreover, his strategy space comprises of actions that yield the highest payoffs.
\item The defender uses a set of contiguous strategies (thresholds on the attack reward)  that consist of the most rewarding attack vectors. 
This is in contrast with known algorithms such
as logistic regression which have a predefined shape of the boundary independently of
the attacker's goal.
\item The weight assigned to each threshold is proportional to the marginal reward increase at that point. If the marginal reward increase is constant, the defender randomizes uniformly among strategies in the interior of her support. 
\end{enumerate}
Note that this equilibrium is calculated easily by
using the results from Section~\ref{sec:analysis}, without any need for a complex program. For the defender, the weight given to each strategy
is constant in the interior of the support, equal to $\beta_i = \frac{r_{i}-r_{i-1}}{c_\text{d}} =  \frac{c_\text{a}}{c_\text{d}} =  \frac{ 1}{120}$.

In the following, we see how changes in the parameters affect the players' equilibrium payoffs. 
First we study how the equilibrium payoffs change as $c_\text{a}$, the cost to the defender of a single attack (or conversely reward to the attacker), increases and $c_d$, the cost to attacker of a detection event, is fixed.
For instance, consider the real world problem of online click fraud which is prominent in pay-per-click (PPC) online advertising. Owners of websites that post the ads are paid an amount of money determined by how many visitors to the sites click on the ads. Malicious owners of sites could choose to use low-priced key words to attract less attention from the fraud classifiers of the ad-network, or be more aggressive and pick high-priced keywords. If the expected reward is very high, the attacker might as well attack with full strength (e.g., choose high-priced keywords) to get the revenue generated by the ads, and risk immediate detection.
As we see in Fig.~\ref{fig:ca}, when $c_\text{a}$ rises, the attacker's strategy becomes
more and more concentrated on attacking the target with more frequency, since the expected cost of being detected becomes relatively less important.
In fact, in the extreme case where $c_\text{a}$ is much larger than $c_\text{d}$, the attacker always attacks the maximum number of times even though he is always detected in the process.
Increasing $c_\text{a}$ relative to $c_\text{d}$ gives the attacker more power to earn reward at the defender's cost, without giving the defender any corresponding increase in the ability to ``fight back.'' Thus one should expect that 
as $c_\text{a}$ rises the equilibrium payoff to the attacker should rise and that of the defender should fall.

\begin{figure}[h!]
  \centering \includegraphics[scale = 0.5]{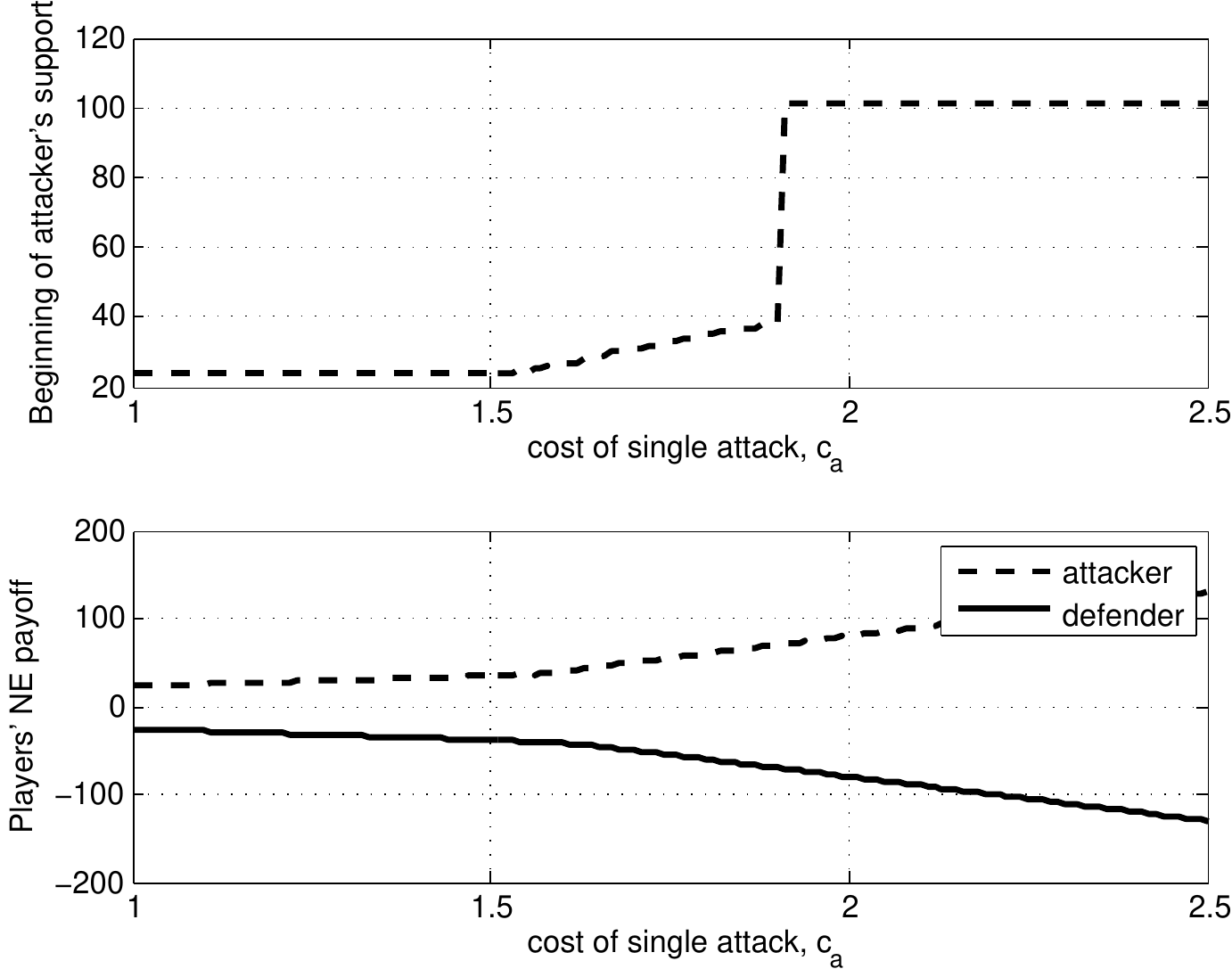}
    \caption{Attacker's support lower bound (top) and equilibrium payoffs of attacker and defender (bottom) as $c_\text{a}$, the cost of a single attack, is varied.}
    \label{fig:ca}
\end{figure}

In Figure~\ref{fig:cf} we observe the impact of the false alarm
penalty parameter. 
In real world scenarios, the false alarm penalty incurred to the defender can differ a lot. For example, credit card companies care more about classifying a real user as fraudster than a mail provider who classifies an email as spam or not spam. As the false alarm penalty increases, the defender reduces false alarms by concentrating her distribution of threshold choice on higher values.
Conversely, the attacker can exploit the higher thresholds by using attack distributions more concentrated on attack strategies that yield higher reward. %
Thus raising $c_\text{fa}$ increases attacker payoff and decreases defender payoff in equilibrium.

\begin{figure}[h!]
  \centering
    \includegraphics[scale=0.5]{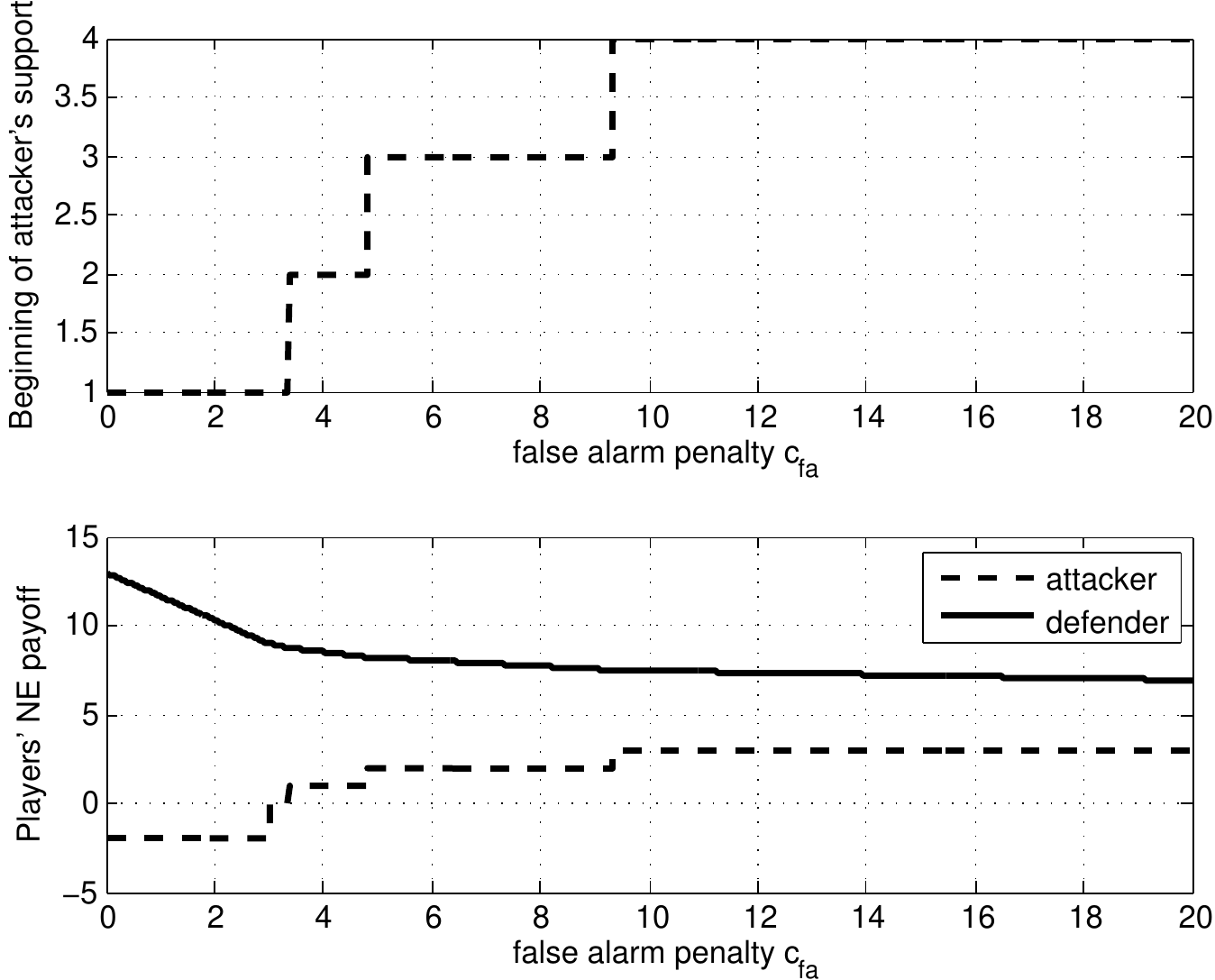}
      \caption{Attacker's support lower bound (top) and equilibrium payoffs of attacker and defender (bottom)  as $c_\text{fa}$, the cost of false alarm, is varied.}
          \label{fig:cf}
\end{figure}

Observe that in Fig.~\ref{fig:cf} the attacker's 
expected payoff has a staircase nature. This is an effect of the discrete optimal values of 
the starting index $s$ of the contiguous block of tight inequality constraints: for some region of parameters, 
the same $s$ is optimal. For the same optimal $s$, the support of the defender's strategy stays the same and the detection benefit, hence the attacker's equilibrium payoff is the same. The defender's staircase nature is slightly affected by the false alarm penalty factor in her payoff function. The above result suggests that in practice, the attacker is not so sensitive to false alarm penalty variations. Thus, even a few different buckets of estimated false alarm penalty costs to the defender would be sufficient for the attacker to compute his optimal strategy. Even if the underlying parameter cost fluctuates by a small amount, the attacker's strategy might remain the same.

\subsection{Multiple Features: optimal investment defender strategies}

In the previous experiments, classification was based on a single feature. In a security environment, one important decision of the defender (or network administrator) is to decide
whether incorporating an additional feature would improve attacker classification, and improve it enough to justify any additional cost in collecting that data.

The feature vector of the defender, upon which the classification is based, may consist of multiple features. For example, the feature vector of Twitter's classification algorithms (which differentiate fraudulent accounts from legitimate ones), might consist of features such as the number of followers, average number of retweets per tweet, number of mentions, country of origin, and others. Before the defender decides to collect more features for her classification purposes, she should be aware that the most successful feature is the one that remains stealthy from the attacker. If she invests a lot into acquiring features that can be easily learned by the attacker, then the window of expected high reward (see Scenario 2 in Fig.~\ref{fig:multiple_features}) might be too small to compensate for the increased expenses of the feature acquisition and maintenance.

In this experiment we suppose that the defender observes which servers in her network a user accesses, and one of these servers is known to all parties to be particularly valuable.  The observation of how many times
this ``valuable server'' is accessed we designate as feature 1. 
A possible second feature the defender can use in classification is inspired by the literature of detecting portscanners~\cite{Jung04a}.
Jung et al. found from empirical data a distinction between benign and malicious portscanners in terms of the proportions of the connections that
were successfully established. In particular they define a  metric, called \textbf{inactive-pct}, as the ratio of the number of hosts to which failed connections are made versus the number of hosts 
to which successful connections are made. 
Jung et al. found that benign users (e.g., web search engines) have a low \textbf{inactive-pct} ($<80\%$), as a larger fraction of connections will be successfully established. 
On the contrary, malicious portscanners often have a high \textbf{inactive-pct} presumably because they initiate a lot of connections to detect vulnerabilities and terminate them immediately.

In our experiment, the attacker decides how many times to access the  ``valuable server'' over a window of $N = 2$ time slots. Moreover, the attacker explores the other ports of the network, in parallel, looking for other
targets of opportunity. He can access these other ports with a low or high \textbf{inactive-pct}. Scanning with a low \textbf{inactive-pct} is less efficient at finding targets, so it is less rewarding for the attacker. In this experiment, the
attacker chooses to attack the ``valuable server'' either 0, 1, or 2 times, and also chooses whether to have a low or high  \textbf{inactive-pct}. Thus there are $3 \times 2 = 6$ attack vectors. The reward for each is set to be $c_\text{a}=1$ times the
number of attacks of the ``valuable server'' plus $c_\text{low}=2$ if \textbf{inactive-pct} is chosen low and $c_\text{high}=4.1$  if \textbf{inactive-pct} is chosen high.
 We also suppose that $p=0.2$ (the prior probability a user is an attacker),  $\theta_0 = 0.3$ (the non-attacker's frequency of access to the ``valuable server''), and $\theta^\text{low} = 0.8$ (the non-attacker's probability to have a low \textbf{inactive-pct}), $c_\text{d} = 1$ (the cost of detection), and $c_\text{fa} = 1$ (the cost of false alarm).
The experiment consists of four scenarios:
\begin{enumerate}[1.]
\item The defender only observes feature 1, the number of times an agent accesses the ``valuable server,'' and the attacker knows only feature 1 is being used in the classifier. Consequently, the attacker only uses the  high  \textbf{inactive-pct} strategy vectors 
since choosing a low  \textbf{inactive-pct} only reduces reward without changing detection probability. 

\item  The attacker continues to play the equilibrium strategy of scenario 1 while the defender now has access to  feature~2,  \textbf{inactive-pct},  and uses this to optimize her classifier. 
The defender assumes, correctly for now, that the attacker keeps the same strategy as in scenario 1. The attacker continues to use the equilibrium strategy of scenario 1 because he does not ``know'' that the defender has access to feature 2.

\item The defender continues to use the equilibrium strategy of scenario 2, but now the attacker knows that the defender changed classifiers to use both features and optimized it against a scenario 1 attacker. 

\item Both features are used by the defender, and it is common knowledge of both players that both features are being used.
\end{enumerate}


As we see in Fig.~\ref{fig:multiple_features}, when the attacker does not know that the defender classifies him based on two features (scenario 2), the defender's payoff increases from scenario 1 in which the defender only classified the agent based on a single feature. When the attacker finds out that he is getting classified on multiple features, then the defender's value decreases from scenario 2, since the attacker is smart enough to realize why he got detected. Comparing scenarios 1 and 4, we see that in this experiment the defender's NE payoff increases when she has access to two features, but the benefit is greatly diminished by the attacker's counter response to the new classifier.

\begin{figure}[t!]
  \centering
  \includegraphics[scale=0.5]{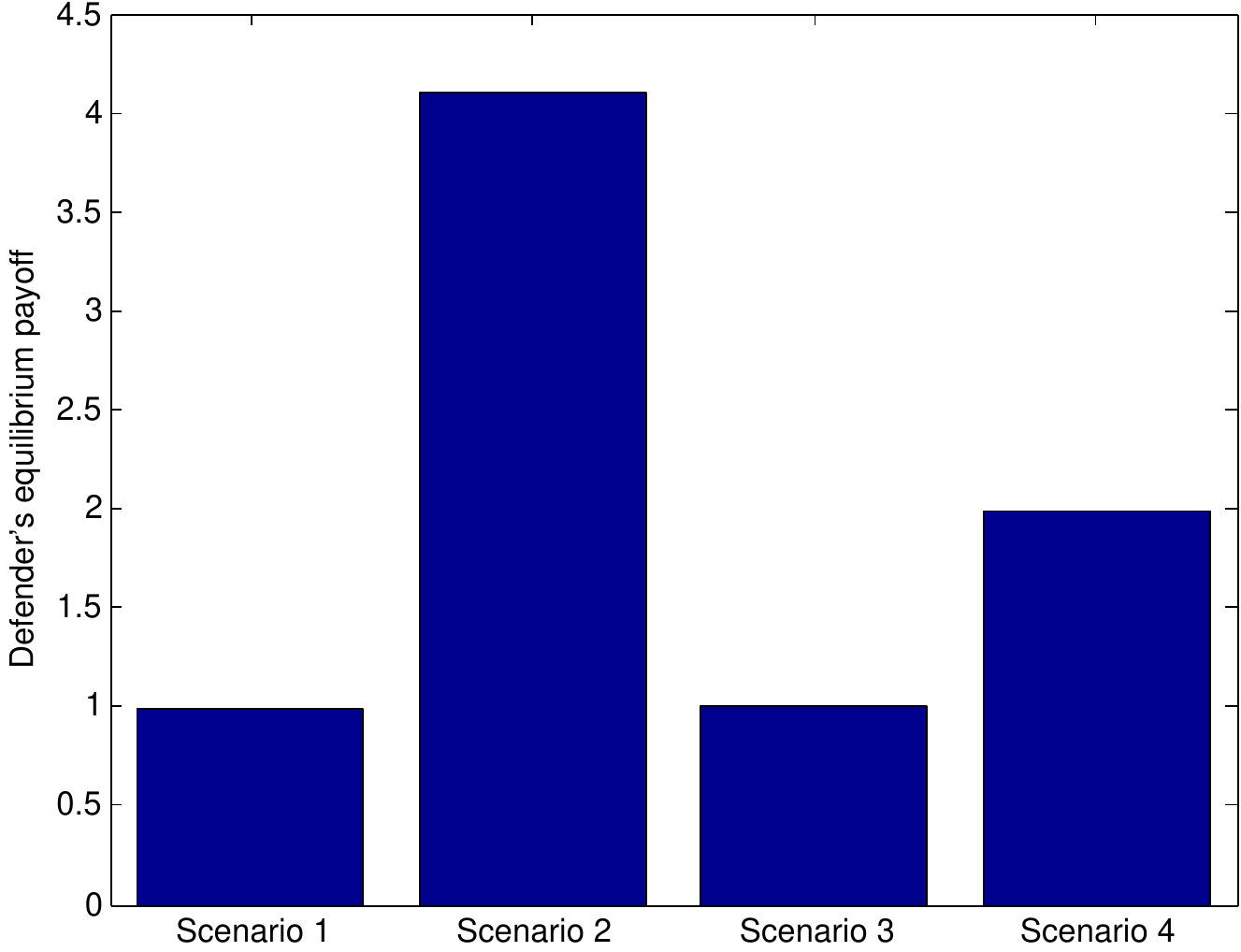}
      \caption{Scenario 1: Defender does not differentiate between low and high ``inactive-pct".
      Scenario 2: Attacker mistakenly believes he is only classified based on a single feature. Scenario 3: Defender mistakenly believes that attacker believes he is classified based on a single feature. Scenario 4: Both players know that two sensors (features) determine classification.} 
       \label{fig:multiple_features}
\end{figure}

\section{Concluding Remarks}
\label{sec:concl}

In this paper, we propose and analyze a new game-theoretic model of adversarial classification. Contrarily to most previous research, our model takes into account the key differences between the objectives of the attacker and defender using a nonzero-sum game, yet it is simple enough to yield analytical results that bring intuitive insights into the structure of the Nash equilibrium and how classification should be performed against this type of adversary. 
In particular we show that, to remain stealthy while maximizing reward, the attacker mixes amongst attack vectors with a distribution proportional to the non-attacker's distribution (i.e., to normal behavior) but on a reduced support of attack vectors with highest rewards. The defender, on the other hand, mixes between classifiers that correspond to applying a threshold on the attacker's reward. This result intuitively shows that, in a strategic setting, the classifier should combine features according to the attacker's reward function (i.e., by looking at the reward a given feature/attack vector gives). Hence, using a standard classifier with a fixed shape of the decision boundary (such as logistic regression with a linear boundary) will necessarily be suboptimal regardless of how the parameters of the model are trained if the reward function does not have the predefined shape. Our results on the defender's equilibrium strategy also show the need for randomization (mixed strategy) and show that the weight assigned to each threshold is mainly proportional to the marginal reward increase at that point.

An important element in the tractability of our model is the special structure of the payoffs that makes our game best-response equivalent to a zero-sum game and hence allows us to use Linear Programming tools to search for all Nash equilibria. Surprisingly, although this equivalence seems straightforward, the literature in security games is largely looking at zero-sum games. We believe that models that better capture realistic scenarios in which the defender and attacker have different tradeoffs in their objective functions but are still computationally equivalent to zero-sum games could be studied using an approach similar to ours. 

A major assumption of most game theoretic models is that players are rational -- simply meaning that each player has preferences over the possible outcomes of the game that can be represented by assigning each outcome a payoff which each player tries to maximize in expectation. In many games such as ours, the solution concept that is most natural is that of a mixed strategy Nash equilibrium. However there are always questions of whether players will actually play these mixed strategies, and indeed this has been an area that the luminaries of game theory have given considerable thought over the years (see Rubinstein \cite{RePEc:ecm:emetrp:v:59:y:1991:i:4:p:909-24} for a review).  As Rubinstein points out, while there are cases in which mixed strategies may be seen as an abstraction of interacting with a population of players, or decision making with private information, there are cases in which players have an incentive to randomize. 
It is the case in our game as both attacker and defender have an incentive to be unpredictable. 
Another domain in which agents are motivated to be unpredictable is sport, where empirical studies have shown players to randomize according to Nash equilibrium distributions~\cite{10.2307/3648639}. The players likely are not calculating the Nash equilibrium analytically but instead are learning about the frequencies in which opponents play certain actions.

Finally, even if one doubts whether the attacker will behave ``rationally,'' the defender has a compelling reason to consider following the mixed strategy prescribed by this work.  The reason is that the defender's equilibrium mixed strategy minimizes the maximum cost the defender can suffer from the attacker, however the attacker chooses to play.  This includes the possibility of an attacker that attacks in some non-strategic way, such as according to an automated script that is not finely tuned to be the optimal attack for a particular defender. On the other hand, a defender might wish to exploit the predictability of some non-strategic attackers by employing a multi-stage approach to first find the attackers who follow a consistent pattern before applying the approach described in this paper. 

Overall, our paper makes a step towards building better attack detection systems using classification techniques that take into account the objective of the attacker to optimize the defense. Our results show that game theory is a valuable tool to tackle adversarial classification problems in settings that are not worst-case as it provides a formal way to justify the need for randomization and to optimize the defense distribution for a given attacker's objective. In future work, we plan to extend our results to more complex adversarial settings. 

\section*{Acknowledgement}
We would like to thank the anonymous reviewers and the associate editor for carefully reading of our manuscript and for providing insightful comments and suggestions.

\bibliographystyle{abbrv}
\bibliography{IEEE_TIFS_2016}

\begin{thebibliography}{10}

\bibitem{Alpcan03a}
T.~Alpcan and T.~Ba\c{s}ar.
\newblock A game theoretic approach to decision and analysis in network
  intrusion detection.
\newblock In {\em Proceedings of IEEE CDC}, pages 2595--2600, 2003.

\bibitem{alpcan2010network}
T.~Alpcan and T.~Ba{\c{s}}ar.
\newblock {\em Network Security: A Decision and Game-Theoretic Approach}.
\newblock Cambridge University Press, 2010.

\bibitem{Inspection_games2002}
R.~Avenhaus, B.~Von~Stengel, and S.~Zamir.
\newblock Chapter 51 inspection games.
\newblock In R.~Aumann and S.~Hart, editors, {\em Handbook of Game Theory with
  Economic Applications}, volume~3, pages 1947--1987. Elsevier, 2002.

\bibitem{Barber12a}
D.~Barber.
\newblock {\em Bayesian Reasoning and Machine Learning}.
\newblock Cambridge University Press, 2012.

\bibitem{Barni13a}
M.~Barni and B.~Tondi.
\newblock {The Source Identification Game: An Information-Theoretic
  Perspective}.
\newblock {\em IEEE Transactions on Information Forensics and Security},
  8(3):450--463, 2013.

\bibitem{Barni14a}
M.~Barni and B.~Tondi.
\newblock Binary hypothesis testing game with training data.
\newblock {\em IEEE Transactions on Information Theory}, 60(8):4848--4866,
  2014.

\bibitem{Barreno10a}
M.~Barreno, B.~Nelson, A.~D. Joseph, and J.~D. Tygar.
\newblock The security of machine learning.
\newblock {\em Machine Learning}, 81(2):121--148, 2010.

\bibitem{Bertsimas:1997:ILO:548834}
D.~Bertsimas and J.~Tsitsiklis.
\newblock {\em Introduction to Linear Optimization}.
\newblock Athena Scientific, 1st edition, 1997.

\bibitem{Boyd}
S.~Boyd and L.~Vandenberghe.
\newblock {\em {Convex Optimization}}.
\newblock Cambridge University Press, 2004.

\bibitem{opac-b1079784}
S.~P. Bradley, A.~C. Hax, and T.~L. Magnanti.
\newblock {\em Applied mathematical programming}.
\newblock Addison-Wesley, Reading Mass., 1977.

\bibitem{Bruckner12a}
M.~Br\"{u}ckner, C.~Kanzow, and T.~Scheffer.
\newblock Static prediction games for adversarial learning problems.
\newblock {\em Journal of Machine Learning Research}, 13:2617--2654, 2012.

\bibitem{Bruckner09a}
M.~Br\"{u}ckner and T.~Scheffer.
\newblock Nash equilibria of static prediction games.
\newblock In {\em Proceedings of NIPS}, pages 171--179, 2009.

\bibitem{Bruckner11a}
M.~Br\"{u}ckner and T.~Scheffer.
\newblock Stackelberg games for adversarial prediction problems.
\newblock In {\em Proceedings of KDD}, pages 547--555, 2011.

\bibitem{DBLP:journals/tifs/ChenL09}
L.~Chen and J.~Leneutre.
\newblock {A Game Theoretical Framework on Intrusion Detection in Heterogeneous
  Networks}.
\newblock {\em IEEE Transactions on Information Forensics and Security},
  4(2):165--178, 2009.

\bibitem{Chen:2009:SCC:1516512.1516516}
X.~Chen, X.~Deng, and S.-H. Teng.
\newblock Settling the complexity of computing two-player nash equilibria.
\newblock {\em J. ACM}, 56(3):14:1--14:57, May 2009.

\bibitem{Christin11_NetworkSecurityGamesCombiningGameTheoryBehavioralEconomics}
N.~Christin.
\newblock {Network Security Games: Combining Game Theory, Behavioral Economics,
  and Network Measurements}.
\newblock In {\em Proceedings of GameSec}, pages 4--6, 2011.

\bibitem{Dalvi04a}
N.~Dalvi, P.~Domingos, Mausam, S.~Sanghai, and D.~Verma.
\newblock Adversarial classification.
\newblock In {\em Proceedings of ACM KDD}, pages 99--108, 2004.

\bibitem{Dresher1962}
M.~Dresher.
\newblock {A sampling Inspection Problem in Arms Control Agreements: a
  Game-theoretic Analysis}.
\newblock In {\em Memorandum RM-2972-ARPA, The RAND Corporation}, 1962.

\bibitem{DLM_cdc_2012}
L.~Dritsoula, P.~Loiseau, and J.~Musacchio.
\newblock {A Game-Theoretical Approach for Finding Optimal Strategies in an
  Intruder Classification Game}.
\newblock In {\em Proceedings of CDC}, 2012.

\bibitem{DLM_gamesec_2012}
L.~Dritsoula, P.~Loiseau, and J.~Musacchio.
\newblock {Computing the Nash Equilibria of Intruder Classification Games}.
\newblock In {\em Proceedings of GameSec}, 2012.

\bibitem{Forges92a}
F.~Forges.
\newblock Chapter 6 repeated games of incomplete information: Non-zero-sum.
\newblock In R.~Aumann and S.~Hart, editors, {\em Handbook of Game Theory with
  Economic Applications}, volume~1, pages 155--177. Elsevier, 1992.

\bibitem{Tirole}
D.~Fudenberg and J.~Tirole.
\newblock {\em {Game Theory}}.
\newblock MIT Press, 1991.

\bibitem{Globerson06a}
A.~Globerson and S.~Roweis.
\newblock Nightmare at test time: Robust learning by feature deletion.
\newblock In {\em Proceedings of ICML}, 2006.

\bibitem{Huang11a}
L.~Huang, A.~D. Joseph, B.~Nelson, B.~I. Rubinstein, and J.~D. Tygar.
\newblock Adversarial machine learning.
\newblock In {\em Proceedings of ACM AISec}, pages 43--58, 2011.

\bibitem{Jung04a}
J.~Jung, V.~Paxson, A.~W. Berger, and H.~Balakrishnan.
\newblock {Fast Portscan Detection Using Sequential Hypothesis Testing}.
\newblock In {\em Proceedings of IEEE S\&P}, 2004.

\bibitem{Karmarkar:1984:NPA:800057.808695}
N.~Karmarkar.
\newblock {A New Polynomial-time Algorithm for Linear Programming}.
\newblock In {\em Proceedings of ACM STOC}, pages 302--311, 1984.

\bibitem{Korzhyk11a}
D.~Korzhyk, Z.~Yin, C.~Kiekintveld, V.~Conitzer, and M.~Tambe.
\newblock Stackelberg vs. nash in security games: An extended investigation of
  interchangeability, equivalence, and uniqueness.
\newblock {\em J. Artif. Int. Res.}, 41(2):297--327, May 2011.

\bibitem{Li15a}
B.~Li and Y.~Vorobeychik.
\newblock Scalable optimization of randomized operational decisions in
  adversarial classification settings.
\newblock In {\em Proceedings of AISTATS}, 2015.

\bibitem{Lisy14a}
V.~Lis{\'y}, R.~Kessl, and T.~Pevn{\'y}.
\newblock Randomized operating point selection in adversarial classification.
\newblock In {\em Proceedings of ECML PKDD}, pages 240--255, 2014.

\bibitem{Liu06a}
Y.~Liu, C.~Comaniciu, and H.~Man.
\newblock A bayesian game approach for intrusion detection in wireless ad hoc
  networks.
\newblock In {\em Proceeding GameNets}, 2006.

\bibitem{Lowd05a}
D.~Lowd and C.~Meek.
\newblock Adversarial learning.
\newblock In {\em Proceedings of ACM KDD}, pages 641--647, 2005.

\bibitem{luenberger2003linear}
D.~Luenberger.
\newblock {\em {Linear and Nonlinear Programming}}.
\newblock Springer, 2nd edition, 2003.

\bibitem{Lunt}
T.~F. Lunt.
\newblock {A Survey of Intrusion Detection Techniques}.
\newblock {\em Computers and Security}, 12(4):405--418, June 1993.

\bibitem{Lye05a}
K.-W. Lye and J.~M. Wing.
\newblock Game strategies in network security.
\newblock {\em Int. J. Inf. Secur.}, 4(1-2):71--86, 2005.

\bibitem{Manshaei13a}
M.~H. Manshaei, Q.~Zhu, T.~Alpcan, T.~Bac\c{s}ar, and J.-P. Hubaux.
\newblock Game theory meets network security and privacy.
\newblock {\em ACM Comput. Surv.}, 45(3):25:1--25:39, 2013.

\bibitem{Maschler1966}
M.~Maschler.
\newblock {A Price Leadership Method for Solving the InspectorÕs Non-constant
  Sum Game}.
\newblock In {\em Naval Research Logistics Quarterly}, 1966.

\bibitem{Nash50}
J.~Nash.
\newblock {\em {Non-Cooperative Games}}.
\newblock PhD thesis, Princeton University, 1950.

\bibitem{nash1996linear}
S.~Nash and A.~Sofer.
\newblock {\em {Linear and Nonlinear Programming}}.
\newblock McGraw-Hill, 1996.

\bibitem{Nelson09a}
B.~Nelson, M.~Barreno, F.~J. Chi, A.~D. Joseph, B.~I.~P. Rubinstein, U.~Saini,
  C.~Sutton, J.~D. Tygar, and K.~Xia.
\newblock Misleading learners: Co-opting your spam filter.
\newblock In P.~S. Yu and J.~J.~P. Tsai, editors, {\em Machine Learning in
  Cyber Trust: Security, Privacy, and Reliability}. Springer, 2009.

\bibitem{Nelson10a}
B.~Nelson, B.~I.~P. Rubinstein, L.~Huang, A.~D. Joseph, S.~Lau, S.~Lee, S.~Rao,
  A.~Tran, and J.~D. Tygar.
\newblock Near optimal evasion of convex-inducing classifiers.
\newblock In {\em Proceedings of AISTATS}, 2010.

\bibitem{10.2307/3648639}
I.~Palacios-Huerta.
\newblock Professionals play minimax.
\newblock {\em The Review of Economic Studies}, 70(2):395--415, 2003.

\bibitem{Rosenthal_1974}
R.~W. Rosenthal.
\newblock {Correlated Equilibria in Some Classes of Two-Person Games}.
\newblock {\em In International Journal of Game Theory}, 3(3):119--128, 1974.

\bibitem{Roy10a}
S.~Roy, C.~Ellis, S.~Shiva, D.~Dasgupta, V.~Shandilya, and Q.~Wu.
\newblock {A Survey of Game Theory as Applied to Network Security.}
\newblock In {\em Proceedings of HICSS}, pages 1--10, 2010.

\bibitem{RePEc:ecm:emetrp:v:59:y:1991:i:4:p:909-24}
A.~Rubinstein.
\newblock {Comments on the Interpretation of Game Theory}.
\newblock {\em Econometrica}, 59(4):909--924, July 1991.

\bibitem{Samusevich16a}
R.~Samusevich.
\newblock Game theoretic optimization of detecting malicious behavior.
\newblock Master's thesis, Czech Technical University in Prague, 2016.

\bibitem{Sebe05a}
N.~Sebe, I.~Cohen, A.~Garg, and T.~S. Huang.
\newblock {\em Machine Learning in Computer Vision}.
\newblock Springer, 2005.

\bibitem{Sierksma02a}
G.~Sierksma.
\newblock {\em Linear and Integer Programming: Theory and Practice}.
\newblock CRC Press, 2nd edition, 2002.

\bibitem{Sommer10outsidethe}
R.~Sommer and V.~Paxson.
\newblock {Outside the Closed World: On Using Machine Learning For Network
  Intrusion Detection}.
\newblock In {\em Proceedings of IEEE S\&P}, 2010.

\bibitem{Stamm12a}
M.~C. Stamm, W.~S. Lin, and K.~J.~R. Liu.
\newblock Forensics vs. anti-forensics: A decision and game theoretic
  framework.
\newblock In {\em Proceedings of ICASSP}, pages 1749--1752, 2012.

\bibitem{stringhini13followers}
G.~Stringhini, G.~Wang, M.~Egele, C.~Kruegel, G.~Vigna, H.~Zheng, and B.~Zhao.
\newblock Follow the green: Growth and dynamics in twitter follower markets.
\newblock In {\em Proceedings of IMC}, 2013.

\bibitem{Tambe11a}
M.~Tambe.
\newblock {\em Security and Game Theory: Algorithms, Deployed Systems, Lessons
  Learned}.
\newblock Cambridge University Press, 2011.

\bibitem{citeulike:1442986}
A.~L. Tarca, V.~J. Carey, X.-W. Chen, R.~Romero, and S.~Dr\u{a}ghici.
\newblock Machine learning and its applications to biology.
\newblock {\em PLoS Comput Biol}, 3(6):e116+, June 2007.

\bibitem{Thomas:2013:TFA:2534766.2534784}
K.~Thomas, D.~McCoy, C.~Grier, A.~Kolcz, and V.~Paxson.
\newblock Trafficking fraudulent accounts: The role of the underground market
  in twitter spam and abuse.
\newblock In {\em Proceedings of USENIX Security}, pages 195--210, 2013.

\bibitem{Tirole_IO}
J.~Tirole.
\newblock {\em {The Theory of Industrial Organization}}.
\newblock The MIT Press, 1988.

\bibitem{Tsai09a}
J.~J.~P. Tsai and P.~S. Yu, editors.
\newblock {\em Machine Learning in Cyber Trust: Security, Privacy, and
  Reliability}.
\newblock Springer, 2009.

\bibitem{Vorobeychik14a}
Y.~Vorobeychik and B.~Li.
\newblock Optimal randomized classification in adversarial settings.
\newblock In {\em Proceedings of AAMAS}, 2014.

\bibitem{Wang14a}
G.~Wang, T.~Wang, H.~Zheng, and B.~Y. Zhao.
\newblock Man vs. machine: Practical adversarial detection of malicious
  crowdsourcing workers.
\newblock In {\em Proceedings of USENIX Security}, pages 239--254, 2014.

\bibitem{Zhou14a}
Y.~Zhou and M.~Kantarcioglu.
\newblock Adversarial learning with bayesian hierarchical mixtures of experts.
\newblock In {\em Proceedings of SIAM SDM}, pages 929--937, 2014.

\bibitem{Zhou12a}
Y.~Zhou, M.~Kantarcioglu, B.~Thuraisingham, and B.~Xi.
\newblock Adversarial support vector machine learning.
\newblock In {\em Proceedings of KDD}, pages 1059--1067, 2012.

\end{thebibliography}
\techrep{
\appendix

\begin{proof}[Lemma~\ref{same_pd}]
For the attacker,  from \eqref{eq:attackerua} we derive
\begin{align*}
U^\text{A}({\vect \alpha, \vect \beta})& = \sum_{v \in \mathcal{V}} \sum_{c \in
  \mathcal{C}} \alpha_v U^\text{A}(v, c) \beta_c \nonumber \\ & =
\sum_{v \in \mathcal{V}} \sum_{c \in \mathcal{C}} \alpha_v \left( R(v)
- c_\text{d} \mathds{1}_ {c(v)=1} \right) \beta_c \nonumber \\ & =
\sum_{v \in \mathcal{V}} \alpha_v \left( R(v) - c_d \sum_{c \in
  \mathcal{C}} \beta_c \mathds{1}_ {c(v)=1} \right) \nonumber \\ & =
\sum_{v \in \mathcal{V}} \Big(\alpha_v R(v) - c_d \alpha_v
\pi^{\vect \beta}_\text{d}(v) \Big),
\end{align*}
where $\pi^{\vect \beta}_\text{d}(v)$ is given by~\eqref{eq:prob_det}.
Similarly, for the defender, \eqref{eq:defenderud} yields

\begin{align*}
U^\text{D}({\vect \alpha, \vect \beta})& = \sum_{v \in \mathcal{V}} \sum_{c \in
  \mathcal{C}} \alpha_v U^\text{D}(v, c) \beta_c \\ 
  & = -\sum_{v \in \mathcal{V}} \sum_{c \in \mathcal{C}} \alpha_v U^\text{A}(v,c)  \beta_c \\
  & ~~~ -   \sum_{c \in \mathcal{C}} \beta_c 
\sum_{v^{\prime} \in \mathcal{V}} \displaystyle \frac{1-p}{p} c_\text{fa} P_\text{N}(v^{\prime})  \mathds{1}_ {c(v^{\prime})=1}\\ 
& = - U^\text{A}({\vect \alpha, \vect \beta}) - \sum_{v^{\prime} \in \mathcal{V}} \displaystyle  \frac{1-p}{p} c_{\text{fa}}
P_\text{N}(v^{\prime}) \pi^{\vect \beta}_\text{d}(v^{\prime}).\qed
\end{align*}
\end{proof}

\begin{proof}[Proposition~\ref{attacker_reduced_prop}]
Let $(\vect \alpha, \vect \beta)$ be a NE of $\mathcal{G}^T$. 

First note that if all attack vectors yield distinct rewards ($|\mathcal{V}| = |\mathcal{V}^R|$), the result holds trivially (and $\vect \alpha^* = \vect \alpha$).

Now assume that there exist $v_1, v_2 \in \mathcal{V} \text{ with } v_1 \neq v_2 \text{ and } R(v_1)= R(v_2)$ and all other vectors yield a distinct reward. Define $r_1 = R(v_1) = R(v_2)$ and $r_i = R(v_i)$ for all $i\in \{3, \cdots, |\mathcal{V}|\}$; and $\vect \alpha^*$ as in Proposition~\ref{attacker_reduced_prop}, that is $\alpha^*_{r_1} = \alpha_{v_1} + \alpha_{v_2}$ and $\alpha^*_{r_i} = \alpha_{v_i}, \forall i\in \{3, \cdots, |\mathcal{V}|\}$. Since $\vect \beta$ is a probability distribution on $\mathcal{C}^T$, the probability of detection is the same for any two attack vectors with the same reward, in particular $\pi_\text{d}(v_1) = \pi_\text{d}(v_2)$. By slight abuse, we denote $\pi_\text{d}(r_1) = \pi_\text{d}(v_1) = \pi_\text{d}(v_2)$ and $\pi_\text{d}(r_i) = \pi_\text{d}(v_i), \forall i\in \{3, \cdots, |\mathcal{V}|\}$. Then, with obvious notation, we have
\begin{align}
\nonumber U^\text{A}(\vect \alpha, \vect \beta) =&   \sum_{i=3}^\mathcal{|V|} (\alpha_{v_i} R(v_i) - c_d \alpha_{v_i} \pi_\text{d}(v_i) ) \\
\nonumber &+ \sum_{i=1}^2 (\alpha_{v_i} R(v_i) - c_d \alpha_{v_i} \pi_\text{d}(v_i)) \\
\nonumber =& \sum_{i=3}^\mathcal{|V|} (\alpha^*_{r_i} r_i - c_d \alpha^*_{r_i} \pi_\text{d}(r_i) ) \\ 
\nonumber &+ \alpha^*_{r_1} r_1 - c_d \alpha^*_{r_1} \pi_\text{d}(r_1) \\
\nonumber =& U^A(\vect \alpha^*, \vect \beta)
\end{align}
and similarly for the defender; and $(\vect \alpha^*, \vect \beta)$ is a NE of $\mathcal{G}^{R, T}$. 

Finally, this reasoning extends straightforwardly to the case for there can be more than two vectors yielding the same reward, hence concluding the proof.
\qed
\end{proof}

\begin{proof}[Lemma \ref{def_lemma}]
By the fundamental theorem of linear programming, if there exists an optimal solution to an LP, it will occur at an extreme point
or a convex combination of extreme points of the polyhedron defined by the constraints. 
 The defender's LP admits an optimal solution, otherwise Nash's theorem of NE existence would be violated. 
 
 Let $[\vect \beta, z]$ be an extreme point of the polyhedron defined by the constraints of
\eqref{def_LP}. We define
$\vect x \triangleq \vect \beta/ z$. Note that $z  = \min[\Lambda \vect \beta]> 0$, since $\Lambda$ is positive, hence 
$\vect x$ is finite. Substituting for $\vect x$, the  constraints in \eqref{def_LP} become
$\Lambda_\text{} \vect x \geq \mathbf{1}_{|\mathcal{V}^R|}$, $\vect x \geq \vect 0$. 
Scaling the objective and the inequality constraints of~\eqref{def_LP} with a positive scalar results in an equivalent
problem which has the same maximizers~\cite[Chapter 4.1.3]{Boyd}.
Thus $\vect x$ is an extreme point of the polyhedron defined by $\Lambda_\text{} \vect x \geq \mathbf{1}_{|\mathcal{V}^R|}$, $\vect x \geq \vect 0$.  This concludes the proof of Lemma~\ref{def_lemma}.
\qed
\end{proof}

\begin{proof}[Lemma~\ref{block_lemma}]
Let $\vect x$ be an extreme point of the polyhedron $P$ defined by \eqref{eq:lp} that corresponds to defender equilibrium strategy $\vect \beta$ with $\vect \beta = \vect x / \| \vect x \|$. Such a point is of dimension $|\mathcal{V}^R| + 1$, thus  $|\mathcal{V}^R| + 1$ constraints are tight at $\vect x$. First note that at least one inequality constraint is tight at $\vect x$. Indeed, otherwise we would have all $|\mathcal{V}^R| + 1$ positivity constraints tight which is impossible since $\vect{0}_{|\mathcal{V}^R| + 1} \notin P$. 
We next show that, if an inequality constraint is tight, then the next one is tight as well. By recursion, that concludes the proof.  

Suppose that, at $\vect x$, inequality constraint $l\in\{1, |\mathcal{V}^R|-1 \}$ is tight and inequality constraint $l+1$ is loose. Subtracting the tight inequality constraint $l$ from the loose inequality constraint $l+1$ yields $x_{l+1}>\displaystyle \frac{r_{l+1} - r_{l}}{c_\text{d}}\|\vect x\|>0$. Hence we can define the transformed vector $\vect{\hat{x}}$ by:
\begin{equation*}
\hat{x}_i = 
\begin{cases}
	 x_i ~~~~~~~~~~~~~\text{for $i \in \{1, \! \cdots\!, l\} \!\cup\! \{l\!+\!3, \cdots\!,   |\mathcal{V}^R|\!+\!1\}$}, \\
         \displaystyle \frac{r_{l+1} - r_l}{c_\text{d}} \|\vect x\|  ~~~~~~~~~~~~~~~~~~~~~~~~~~~~~\text{for $i= l+1$},  \\ 
          x_{l+2} + x_{l+1} -  \displaystyle \frac{r_{l+1} - r_l}{c_\text{d}} \|\vect x\|   ~~~~~~~~~~~\text{for $i= l+2$.}
\end{cases}
\end{equation*}
With this definition, we have $ \|\hat{\vect x}\| = \|\vect x\|$ and, since both $\vect{x}$ and $\vect{\hat{x}}$ belong to $P$ and have (at least) a tight inequality constraint, $\min[\Lambda \vect{x}]=\min[\Lambda \vect{\hat{x}}]=1$. Hence, defining $\vect{\hat{\beta}} = \vect{\hat{x}} / \| \vect{\hat{x}} \|$, we have $\min[\Lambda \vect{\hat{\beta}}]= \min[\Lambda \vect{\beta}]$. Finally, since $\vect \mu$ is strictly decreasing by assumption, we have $\vect \mu' (\vect \beta - \hat{\vect \beta}) = \mu_{l+1} (\beta_{l+1}-\hat{\beta}_{l+1}) + \mu_{l+2} (\beta_{l+2} - \hat{\beta}_{l+2}) >0$.
We conclude that 
$$\min[\Lambda \vect{\hat{\beta}}] - \mu^{\prime} \vect{\hat{\beta}}  > \min[\Lambda \vect{\beta}] - \mu^{\prime} \vect{\beta},$$
which contradicts the fact that $\vect \beta$ is a NE strategy.
\qed
\end{proof}

\begin{proof}[Lemma~\ref{forms_lemma}]
Let $\vect x$ be an extreme point of polyhedron $P$ that corresponds to a defender NE strategy $\vect \beta$. By Lemma~\ref{block_lemma}, inequality constraints $s$ through $|\mathcal{V}^R|$ are tight at $\vect x$.

We first show that $x_i = 0, \forall i \in \{1, \cdots, s-1 \}$. Suppose that there exists $i \in  \{1, \cdots, s-1 \}$ s.t. $x_i > 0$. From $\vect x$, we define a new point $\vect{\hat{x}}$ by the following transformation. Start with $\vect{\hat{x}} = \vect x$, then reduce $\hat{x}_i$ and increase $\hat{x}_{i+1}$ by the same amount until either $\hat{x}_i = 0$ or the inequality constraint $i$ is tight. 
If $x_{i+1}>0$, then the tight constraints at $\vect x$ remain tight at $\vect{\hat{x}}$ but $\vect{\hat{x}}$ has one more tight constraint (either positivity or inequality) which contradicts the fact that $\vect x$ is an extreme point. 
If $x_{i+1}=0$, then $\vect x$ and $\vect{\hat{x}}$ have the same number of tight constraints, but defining $\vect{\hat{\beta}} = \vect{\hat{x}} / \| \vect{\hat{x}} \|$ we have $\min[\Lambda \vect{\hat{\beta}}]= \min[\Lambda \vect{\beta}]$ and $\vect \mu' \vect \beta > \vect \mu'  \vect{\hat{\beta}}$, which contradicts the fact that  $\vect \beta$ is a NE strategy.

Next, we note that for any $i \in \{ s+1, \cdots, |\mathcal{V}^R|\}$, subtracting tight inequality $i-1$ from tight inequality $i$ yields $x_i = \displaystyle \frac{r_i-r_{i-1}}{c_\text{d}} \| \vect x\|>0$.

 We finally show that $x_s$ and $x_{|\mathcal{V}^R|+1}$ cannot be both positive. Suppose that $x_s>0$ and $x_{|\mathcal{V}^R|+1} > 0$. Then, we define $\vect{\hat{x}}$ from $\vect x$ by the following transformation:
\begin{equation}
\hat{x}_i =
\begin{cases}
	 \gamma x_i & \forall i \in \{1, \ldots, s-1\} \cup \{s+1,
         \ldots, |\mathcal{V}^R|\} \\ 0 & \text{for $i= s$} \\ \gamma (x_s + x_i) &
         \text{for $i= |\mathcal{V}^R|+1$,}
\end{cases}
\end{equation} with $\gamma = \displaystyle \frac{1}{1-c_\text{d} x_s} = \displaystyle \frac{1}{\|\vect x\| [ r_{|\mathcal{V}^R|} - r_s + \epsilon]}>1$. The definition of $\gamma$ is such that the $s^\text{th}$ inequality constraint is still tight after the transformation: 
$\gamma \left( c_\text{d} 0 + [r_{|\mathcal{V}^R|} - r_s + \epsilon] \| \vect{x}\|
\right) = 1.$ The loose inequality constraints before $s$ become looser at $\vect{\hat{x}}$ and the previously tight inequality constraints $s+1$ through $|\mathcal{V}^R|$ remain tight at $\vect{\hat{x}}$; but there is one extra tight positivity constraint at $\vect{\hat{x}}$ ($\hat{x}_s=0$), which contradicts the fact that $\vect x$ is an extreme point. 

Combining the above three arguments, we conclude that $\vect x$ has a form given by one of the two types in the lemma.

We now show that there is at most one extreme point $\vect x$ of Type I that corresponds to a NE strategy $\vect \beta$. Suppose there exist two extreme points $\vect x^{1a}, \vect x^{1b}$ of Type I that both correspond to NE defender strategies $\vect \beta^{1a}, \vect \beta^{1b}$. Then,  $$\min[\Lambda \vect \beta^{1a}] - \vect \mu^{\prime} \vect \beta^{1a} = \min[\Lambda \vect \beta^{1b}] - \vect \mu^{\prime} \vect \beta^{1b}.$$ 
Without loss of generality, suppose that $s_{1b} > s_{1a}$.  The inequality constraint with index $|\mathcal{V}^R|$ is tight at both extreme points, therefore $\| \vect x^1 \|= \|\vect x^2 \| = \displaystyle \frac{1}{c_\text{d}+ \epsilon}$ which yields that $ \min[\Lambda \vect \beta^{1a}]  =\min[\Lambda \vect \beta^{1b}] = c_\text{d} + \epsilon$. Then we have 
\begin{align*}
 \min[\Lambda \vect \beta^{1b}] - \vect \mu^{\prime} \vect{\beta}^{1b}  &= \min[\Lambda \vect \beta^{1a}] - \vect \mu^{\prime} \vect \beta^{1b} \\
							&> \min[\Lambda \vect \beta^{1a}] - \vect \mu^{\prime} \vect \beta^{1a},	
							\end{align*}
where the last inequality is a result of the strictly decreasing false alarm penalty vector $\vect \mu$; which is a contradiction. 

Finally, we show that there exist at most two adjacent optimal extreme points of Type II. 
Suppose that there exist three optimal extreme points $\vect x^{1}, \vect x^2, \vect x^3$ of Type II that correspond to defender NE strategies $\vect \beta^{1}, \vect \beta^2, \vect \beta^3$ respectively. Then these are adjacent~\cite{Bertsimas:1997:ILO:548834}.  Without loss of generality, suppose that $s_{2}, s_{2}+1, s_2+2$ are the starting indices of the block of inequality constraints that are tight at $\vect x^{1}, \vect x^2, \vect x^3$ respectively. 
Since all three extreme points are optimal, we have 
\begin{equation}
 \min[\Lambda \vect \beta^{1}]  \! - \! \vect \mu^{\prime} \vect{\beta}^{1}  \! =\!   \min[\Lambda \vect \beta^{2}]\!  - \!  \vect \mu^{\prime} \vect{\beta}^{2} \! =  \! \min[\Lambda \vect \beta^{3}] \! - \!  \vect \mu^{\prime} \vect{\beta}^{3}.
 \label{eq:tie_three}
\end{equation}
The first inequality constraint that is tight for each extreme point yields:
$$
 \min[\Lambda \vect \beta^{i}] = r_{|\mathcal{V}^R|} - r_{s_2 + i-1} + \epsilon, \text{ for } 1\le i \le 3.$$
Subtracting consecutive tight inequalities yield
$$ \beta^i_{j} = \frac{r_{j} - r_{j - 1}}{c_\text{d}},\text{ for } 1\le i \le 3,  ~s_2 + i - 1\le j \le  |\mathcal{V}^R|,$$
and since all three extreme points are of Type II, we have
\begin{equation*} 
\beta^i_{j} = 0,  \text{ for } 1\le i \le 3,  ~0\le j \le s_2 + i -1; 
\end{equation*}
With this, the first and second equality in \eqref{eq:tie_three} bring
\begin{align*}
c_\text{d} = & \mu_{s_2+1} - \mu_{|\mathcal{V}^R|+1}, \\
c_\text{d} = & \mu_{s_2+2} - \mu_{|\mathcal{V}^R|+1}, 
\end{align*}
respectively. This implies $\mu_{s_2+1} = \mu_{s_2+2}$, which contradicts the assumption that $\vect \mu$ is a strictly decreasing vector. 
The proof extends straightforwardly to rule out cases with more than three optimal extreme points. 
\end{proof}

\begin{proof}[Theorem~\ref{aggregate_theorem}]
Since game $\mathcal{G}^{R,T}$ is best-response equivalent to the zero-sum game with payoff matrix $\Lambda^{\textrm{eq}}$ for the defender, finding all Nash equilibria of $\mathcal{G}^{R,T}$ is equivalent to finding all solutions $\vect \beta$ of the defender's LP~\eqref{def_LP} and all solutions $\vect \alpha$ of the attacker's LP which is the dual of \eqref{def_LP}:
\begin{equation}
  \begin{aligned}
    & \underset{\vect \alpha, y}{\text{maximize}} & & y\\ 
    & \text{subject to} & & \vect \alpha^{\prime} \Lambda + y  \mathbf{1}_{|\mathcal{V}^R|+1} \leq  \vect \mu^{\prime}\\ 
    & & & \mathbf{1}^{\prime}_{|\mathcal{V}^R|} \cdot \vect{\alpha} \ge 1, ~\vect \alpha \geq \vect 0.\\
  \end{aligned}
  \label{att_LP}
\end{equation}
To find all the solutions of \eqref{def_LP}, we will use Lemma~\ref{forms_lemma} (recall that $\vect \beta = \vect x/ \| \vect x\|$). From the defender's solution $\vect \beta$, we will be able to give all the attacker's solutions $\vect \alpha$ using the complementary slackness condition of linear programming ~\cite[Chapter 6.2.1]{nash1996linear} and Lemma~\ref{scaled_spam}. 
Before enumerating the possible solutions, let us observe that the $j$-th inequality constraint of the dual \eqref{att_LP} can be written as 
\begin{align*}
c_\text{d} \sum_{i=j}^{|\mathcal{V}^R|} \alpha_i  - \sum_{i=1}^{|\mathcal{V}^R|} r_i \alpha_i + r_{|\mathcal{V}^R|} + \epsilon + y \leq \mu_{j}.
\end{align*}
If two consecutive inequality constraints $j$ and $j+1$ of the dual are tight, subtracting the second from the first gives $\alpha_j = \frac{\mu_j - \mu_{j+1}}{c_{\text{d}}} = \frac{1-p}{p} \frac{c_{\text{fa}}}{c_\text{d}} P_N^R(r_j)$. By complementary slackness, this will happen with $\beta_j>0$ and $\beta_{j+1}>0$, which implies $\pi_\text{d}(r_i) \in (0, 1)$. Hence we recover the result of Lemma~\ref{scaled_spam} using complementary slackness.

From Lemma~\ref{forms_lemma}, there can be at most three optimal basic feasible solutions of \eqref{def_LP}, one of type I and two adjacent of type II. We call $U^D_1$ (resp. $U^D_2$) the largest defender's utility for a basic feasible solution of type I (resp. type II). Each optimal basic feasible solution can be degenerate or non-degenerate: a type I solution is degenerate is $\beta_{s_1}=0$ and a type II solution is degenerate if $\beta_{|\mathcal{V}^R|+1}=0$. To prove Theorem~\ref{aggregate_theorem}, we simply go exhaustively through all possible cases.
\begin{trivlist}

\item[Case 1:] there exists a unique optimal basic feasible solution of type I ($U^D_1>U^D_2$) with index $s_1=s_1^*$. Then, we have $\beta_{s_1^*} \neq 0$ and  $\beta_{s_1^*} \neq k$ (otherwise the solution would also be of type II), $\beta_{s_1^*} \leq \frac{r_k - r_{k-1}}{c_\text{d}}$ (otherwise the $k$-th inequality constraint of the primal is violated), and $\beta_{|\mathcal{V}^R|+1} = 0$. Hence $\vect \beta$ is of the form of Theorem~\ref{aggregate_theorem}-($i$) with $k = s_1^*$. 

Since this solution is unique and non-degenerate, the dual also has a unique and non-degenerate solution. By complementary slackness, we have $\alpha_i = 0, \forall i \in \{1, \ldots, s_1^*-1\}$ and by Lemma~\ref{scaled_spam}, $\alpha_i$ satisfies  \eqref{eq:att_weight} for all $i \in \{ s_1^*, \cdots,  |\mathcal{V}^R| - 1\}$ since $\pi_\text{d}(r_i) \in (0, 1)$. Finally, $\alpha_{|\mathcal{V}^R|}>0$, otherwise the solution would be degenerate. Hence $\vect \alpha$ is also of the form of Theorem~\ref{aggregate_theorem}-($i$) with $k = s_1^*$.

\item[Case 2:] there exists a unique optimal basic feasible solution of type II ($U^D_2 > U^D_1$) with index $s_2=s_2^*$. Then, we have $\beta_{s_2^*} = 0$ and  $\beta_{|\mathcal{V}^R|+1}>0$. Hence $\vect \beta$ is of the form of Theorem~\ref{aggregate_theorem}-($ii$) with $k = s_2^*$. 

Since this solution is unique and non-degenerate, the dual also has a unique and non-degenerate solution. By complementary slackness, we have $\alpha_i = 0, \forall i \in \{1, \ldots, s_2^*-1\}$ and by Lemma~\ref{scaled_spam} or complementary slackness, $\alpha_i$ satisfies  \eqref{eq:att_weight} for all $i \in \{ s_2^*+1, \cdots,  |\mathcal{V}^R|\}$. Finally, $\alpha_{s_2^*}\neq 0$ and $\alpha_{s_2^*} \neq \frac{1-p}{p} \frac{c_{\text{fa}}}{c_\text{d}} P_N^R(r_{s_2^*})$ (otherwise the solution would be degenerate) and $\alpha_{s_2^*} \leq \frac{1-p}{p} \frac{c_{\text{fa}}}{c_\text{d}} P_N^R(r_{s_2^*})$ (otherwise the $s_2^*$-th inequality constraint of the dual would be violated). Hence $\vect \alpha$ is also of the form of Theorem~\ref{aggregate_theorem}-($ii$) with $k = s_2^*$.

\item[Case 3:] there exists a unique degenerate optimal basic feasible solution which is both of type I and II ($U^D_2 = U^D_1$) with index $s_1=s_2=s^*$. Then, we have $\beta_{s^*} = 0$ and $\beta_{|\mathcal{V}^R|+1}=0$. Hence $\vect \beta$ is of the form of Theorem~\ref{aggregate_theorem}-($iii$) with $k = s^*$. 

Since this solution is unique and degenerate, the dual has multiple solutions. By complementary slackness, we have $\alpha_i = 0, \forall i \in \{1, \ldots, s^*-1\}$ and by Lemma~\ref{scaled_spam} or complementary slackness, $\alpha_i$ satisfies  \eqref{eq:att_weight} for all $i \in \{ s^*+1, \cdots,  |\mathcal{V}^R|-1\}$. We also have $\alpha_{s^*} \leq \frac{1-p}{p} \frac{c_{\text{fa}}}{c_\text{d}} P_N^R(r_{s^*})$ (otherwise the $s^*$-th inequality constraint of the dual would be violated). Hence, the only two possible basic feasible solutions for the dual are the one with $\alpha_k=0$ and the one with $\alpha_k = \min (\frac{1-p}{p} \frac{c_{\text{fa}}}{c_\text{d}} P_N^R(r_k),  1-\sum_{i = k+1}^{|\mathcal{V}^R|-1}  \alpha_i)$, and with $\alpha_{|\mathcal{V}^R|} = 1-\sum_{i = k}^{|\mathcal{V}^R|-1}  \alpha_i$ in each case (note that if $\frac{1-p}{p} \frac{c_{\text{fa}}}{c_\text{d}} P_N^R(r_k) =  1-\sum_{i = k+1}^{|\mathcal{V}^R|-1}  \alpha_i$, then the second solution has $\alpha_{|\mathcal{V}^R|} = 0$ and is degenerate but this is no contradiction). Any linear combination of these two basic feasible solutions is optimal. Hence $\vect \alpha$ is also of the form of Theorem~\ref{aggregate_theorem}-($iii$) with $k = s^*$.

\item[Case 4:] there exist two non-degenerate optimal basic feasible solutions of type II ($U^D_2 > U^D_1$) with indices $s_2=s_{2a}^*$ and $s_2=s_{2b}^* = s_{2a}^*+1$. Then, any linear combination of these two solutions is an optimal solution and has $\beta_{s_{2a}^*}=0$, $\beta_{s_{2a}^*+1} \in [0, \frac{r_{s_{2a}^*+1} - r_{s_{2a}^*}}{c_\text{d}}]$ and $\beta_{|\mathcal{V}^R|+1}>0$ (both basic feasible solutions have $\beta_{|\mathcal{V}^R|+1}>0$ since they are non-degenerate). Hence $\vect \beta$ is of the form of Theorem~\ref{aggregate_theorem}-($iv$) with $k = s_{2a}^*+1$. 

Since there are multiple non-degenerate solutions for the primal, there is a unique degenerate solution for the dual. By complementary slackness, we have $\alpha_i = 0, \forall i \in \{1, \ldots, s_{2a}^*\}$ and by Lemma~\ref{scaled_spam} or complementary slackness, $\alpha_i$ satisfies \eqref{eq:att_weight} for all $i \in \{ s_{2a}^*+1, \cdots,  |\mathcal{V}^R|\}$. Hence $\vect \alpha$ is also of the form of Theorem~\ref{aggregate_theorem}-($iv$) with $k = s_{2a}^*+1$.

\item[Case 5:] there exist one non-degenerate optimal basic feasible solution of type I and one non-degenerate optimal basic feasible solution of type II ($U^D_2 = U^D_1$) with indices $s_1=s_{1}^*$ and $s_2=s_{2}^*$. These two solutions must be adjacent, hence $s_{1}^*=s_{2}^*$. For arguments similar to case 1, the type I solution has $\beta_{s_1^*}\in (0, \frac{r_k - r_{k-1}}{c_\text{d}})$ and $\beta_{|\mathcal{V}^R|+1}=0$ whereas for arguments similar to case 2, the type II solution has $\beta_{s_1^*}=0$ and $\beta_{|\mathcal{V}^R|+1}>0$. Any linear combination of these two basic feasible solutions is an optimal solution. Hence $\vect \beta$ is of the form of Theorem~\ref{aggregate_theorem}-($iv$) with $k = s_{1}^*$. 

Since there are multiple non-degenerate solutions for the primal, there is a unique degenerate solution for the dual.  By complementary slackness, we have $\alpha_i = 0, \forall i \in \{1, \ldots, s_{1}^*-1\}$ and by Lemma~\ref{scaled_spam} or complementary slackness, $\alpha_i$ satisfies \eqref{eq:att_weight} for all $i \in \{ s_{1}^*, \cdots,  |\mathcal{V}^R|\}$. Hence $\vect \alpha$ is also of the form of Theorem~\ref{aggregate_theorem}-($iv$) with $k = s_{1}^*$.

\item[Case 6:] there exist one degenerate optimal basic feasible solution of type I and type II with indices $s_1=s_{1}^*=s_2 = s_{2a}^*$ and one non-degenerate optimal basic feasible solution of type II ($U^D_2 = U^D_1$) with index $s_2=s_{2b}^*=s_{2a}^*+1$. The first solution has $\beta_{s_{2a}^*} = 0$ and $\beta_{|\mathcal{V}^R|+1}=0$. The second solution has $\beta_{s_{2a}^*+1} = 0$ and $\beta_{|\mathcal{V}^R|+1}>0$. Any linear combination of these two basic feasible solutions is an optimal solution. Hence $\vect \beta$ is of the form of Theorem~\ref{aggregate_theorem}-($iv$) with $k = s_{2a}^*+1$. 

Since there are multiple solutions for the primal, any solution for the dual is degenerate and since at least one of the primal's solution is non-degenerate, the dual's solution must be unique (otherwise, the dual's dual which is the primal would have all its solutions degenerate). By complementary slackness and Lemma~\ref{scaled_spam} applied to both of the primal's optimal basic feasible solutions, the dual solution must have $\alpha_i = 0, \forall i \in \{1, \ldots, s_{2a}^*\}$ and $\alpha_i$ satisfies \eqref{eq:att_weight} for all $i \in \{ s_{2a}^*+1, \cdots,  |\mathcal{V}^R| -1\}$. Hence $\vect \alpha$ is also of the form of Theorem~\ref{aggregate_theorem}-($iv$) with $k = s_{2a}^*+1$. 

\item[Case 7:] there exist one non-degenerate optimal basic feasible solution of type I and two non-degenerate optimal basic feasible solutions of type II ($U^D_2 = U^D_1$) with indices $s_1=s_{1}^*$, $s_2=s_{2a}^*=s_1^*$ and $s_2=s_{2a}^*+1$. 

This case cannot happen. Indeed, there would be a unique non-degenerate solution for the dual. By complementary slackness or Lemma~\ref{scaled_spam} applied to the type I solution, we would have that $\alpha_{s_{2a}^*}$ satisfies \eqref{eq:att_weight} hence $\alpha_{s_{2a}^*}>0$ and by complementary slackness applied to the second type II solution, we would have $\alpha=0$, which is impossible. 
\qed 
\end{trivlist}
\end{proof}
}{}

\end{document}